\documentclass[preprint,12pt]{elsarticle}
\pdfoutput=1
\usepackage{amsmath}
\usepackage{amsfonts}
\usepackage{latexsym}
\usepackage{amscd}
\usepackage{amssymb}
\usepackage{graphicx}
\usepackage{ctable}
\usepackage{placeins}
\usepackage{float}
\usepackage{tabularx}
\usepackage{color}
\usepackage{subfigure}
\usepackage{graphics} 
\usepackage{graphicx} 
\usepackage{epsf}
\usepackage{epsfig}
\makeatletter
\def\ps@pprintTitle{%
 \let\@oddhead\@empty
 \let\@evenhead\@empty
 \def\@oddfoot{}%
 \let\@evenfoot\@oddfoot}
\makeatother
\usepackage{cancel}
\usepackage{float}
\floatstyle{plaintop}
\restylefloat{table}
\usepackage{hyperref}

\restylefloat{float}

\allowdisplaybreaks[1]

\newtheorem{proof}{Proof}
\newtheorem{proposition}{Proposition}

\journal{Journal of Theoretical Biology}

\title{Modeling Contact Tracing in Outbreaks with Application to Ebola}

\begin{document}

\begin{frontmatter}



\title{Modeling Contact Tracing in Outbreaks with Application to Ebola}


\author[Vandy]{Cameron Browne\corref{corAut}}
\author[GT,GTm]{Hayriye Gulbudak}
\author[Vandy]{Glenn Webb}
\address[Vandy]{Department of Mathematics, Vanderbilt University}
\address[GT]{School of Biology, Georgia Institute of Technology}
\address[GTm]{School of Mathematics, Georgia Institute of Technology}
\cortext[corAut]{corresponding author: cambrowne@louisiana.edu}


\begin{abstract}
Contact tracing is an important control strategy for containing Ebola epidemics.   From a theoretical perspective, explicitly incorporating contact tracing with disease dynamics presents challenges, and population level effects of contact tracing are difficult to determine.  In this work, we formulate and analyze a mechanistic SEIR type outbreak model which considers the key features of contact tracing, and we characterize the impact of contact tracing on the effective reproduction number, $\mathcal R_e$, of Ebola.  In particular, we determine how relevant epidemiological properties such as incubation period, infectious period and case reporting, along with varying monitoring protocols, affect the efficacy of contact tracing.  In the special cases of either perfect monitoring of traced cases or perfect reporting of all cases, we derive simple formulae for the critical proportion of contacts that need to be traced in order to bring the effective reproduction number $\mathcal R_e$ below one.  Also, in either case, we show that $\mathcal R_e$ can be expressed completely in terms of \emph{observable reported case/tracing quantities}, namely $\mathcal R_e=k\dfrac{(1-q)}{q}+k_m$ where $k$ is the number of secondary traced infected contacts per primary \emph{untraced} reported case,  $k_m$ is the number of secondary traced infected contacts per primary \emph{traced} reported case and $(1-q)/q$ is the odds that a reported case is not a traced contact.  These formulae quantify contact tracing as both an intervention strategy that impacts disease spread and a probe into the current epidemic status at the population level.  Data from the West Africa Ebola outbreak is utilized to form real-time estimates of $\mathcal R_e$, and inform our projections of the impact of contact tracing, and other control measures, on the epidemic trajectory.  
\end{abstract}


\end{frontmatter}

\section{Introduction}

  Contact tracing has recently gained public attention because of its importance as a control strategy in the 2014-2015 Ebola outbreaks.   Contact tracing is a potentially powerful disease control strategy in which the close contacts of reported/isolated cases are traced and monitored so that if they become symptomatic they can be efficiently isolated and, in turn, cause reduced transmissions. There are numerous studies about contact tracing (also known as active case finding) over the past two decades.  However, from a theoretical point of view, explicit modeling of the contact tracing process has been challenging due to the complexities in reconciling forward dynamics of an epidemic with the action of tracers working through a transmission chain, along with accounting for particular disease characteristics and public health capabilities.  Thus, formulation of mechanistic models of contact tracing that are analytically tractable and epidemiologically relevant is an important problem in epidemiology. 
  
 Because Ebola is typically transmitted through direct contact with bodily fluids of a symptomatic individual and the incubation period is relatively long (average around $11$ days), contacts that have been exposed to Ebola virus can be identified, monitored, and, when symptomatic, be isolated to limit spread \cite{Team, Fraser}.  Previous outbreaks have been rapidly controlled with contact tracing and isolation, along with limiting hospital and funeral transmission \cite{faye2015chains}.  However, the failure of initial containment and the subsequent unprecedented scale of the West Africa Ebola outbreak in 2014-2015 have challenged public health authorities to employ effective control measures.  As of June 2015, there have been over 27,000 reported cases in the outbreak, however the incidence has decreased markedly over the past six months due to enhanced control efforts \cite{team2015west}.  Therefore it is vital to evaluate how different interventions, such as contact tracing, have affected the epidemic and can bring the outbreak to an end.    Direct measurement of the impact of contact tracing on an outbreak has recently been considered in relation to tuberculosis \cite{guzzetta2015effectiveness}, but a general modeling framework for measuring the population-level effect of contact tracing with epidemic data is needed for Ebola and other emerging pathogens.  
 
 Previous works about contact tracing have used a range of modeling methods from individual based models on specific networks to compartmental ordinary differential equations at the population level.  While networks are a natural setting for contact tracing \cite{Kiss}, their connection with traditional epidemic compartmental models (through a mean-field approximation) should be emphasized.  Therefore it is desirable to have a modeling framework which can work in both settings with sufficient detail, mechanism and simplicity. Also, branching process models \cite{Ball, Muller,Klinkenberg} have been utilized for explicit contact tracing structure, but lack both the simplicity and compartmental detail inherent in the differential equation model we formulate in this article.   Many differential equation models have incorporated contact tracing implicitly \cite{Rivers, Mubayi}, though  Hyman et al. \cite{Hyman} considered a deterministic model with explicit contact tracing for HIV, a different setting than an emerging outbreak such as Ebola.    Mathematical modeling of Ebola has been considered by many authors, e.g. \cite{althaus2014,Team, Chowell,camacho2015,chowell2004,chowell2014west,pmid25685633,Fisman,merler2015spatiotemporal,Kiskowski,kupferschmidt2014estimating,legrand2007understanding,Towers,Weitz}, and several of Ebola modeling studies have implicitly incorporated effects of contact tracing (along with other interventions) \cite{Rivers, Pandey,chowell2015modelling,chowell2004,WebbCT}.  However implicit inclusion of contact tracing fails to capture the true effect of tracing on the effective reproduction number $\mathcal R_e$, a key quantity of interest for epidemiologists.  
 
   In this work, we develop a deterministic outbreak model of contact tracing particularly applicable to Ebola epidemics, which explicitly links tracing back to transmissions, and incorporates disease traits and control together with monitoring protocols.  We calculate the effective reproduction number, $\mathcal R_e$, and simulate the model under different control scenarios.  In particular, we perform sensitivity analysis to determine how disease characteristics, variable monitoring protocols, case reporting and timing of intervention affect the efficacy of contact tracing in controlling the epidemic.  In the special case where either traced contacts are always prevented from causing further infections (perfect monitoring) or 100\% of cases are reported (perfect reporting), we derive simple formulae for the critical proportion of contacts that need to be traced in order to bring $\mathcal R_e$ below one.  
   
   Furthermore, from the explicit linking of tracing and transmission in our model, we derive novel formulae directly relating the effective reproduction number $\mathcal R_e$ to contact tracing observables.  Indeed, define $k$ as the average number of secondary infected contacts traced per primary reported \emph{untraced} case, and $q$ as the fraction of reported cases which are traced contacts.  Then, in the case of \emph{perfect monitoring} of traced contacts and the general setting of unknown underreporting with distinct reported/unreported case infectious periods, we show that the following formula holds:
 \begin{align*}
 \mathcal R_e&= k\left( \frac{1-q}{q} \right)
 \end{align*}
 Thus, $\mathcal R_e$ is simply the product of the average number of secondary infected traced contacts per untraced reported case and the odds that a reported case is not a traced contact.   If contact monitoring is \emph{imperfect}, but case reporting is assumed to be perfect (100\% of cases reported), then $\mathcal R_e$ can also be expressed simply in terms of contact tracing observables.  In this instance, define $k_m$ as the average number of secondary infected contacts traced per primary reported \emph{traced} case, then:
  \begin{align*}
 \mathcal R_e&= k\left( \frac{1-q}{q} \right) + k_m
 \end{align*}
  Utilizing current data and this formula, we form weekly estimates of $\mathcal R_e$ for the Ebola outbreaks in Sierra Leone and Guinea, with the main goal of determining the impact of contact tracing on the epidemics.   Taken together, our model and results quantify contact tracing as both a dynamic intervention strategy impacting disease spread and a probe into the current epidemic status at the population level.

\vspace{.5cm}

\section{Model and Methods}
\subsection{Base SEIR model}
We begin with an SEIR-type base model consisting of compartments representing susceptible ($S$), exposed or incubating ($E$), and two distinct infectious groups, namely infectious individuals which will be hospitalized/reported ($I_h$) and infectious individuals which will not be hospitalized and unreported ($I_u$), along with the decoupled compartment of cumulative hospitalized/reported cases ($H$).  The following differential equation system models the dynamics of these populations:
\begin{align*}
S'(t)&=-\beta S(t) \left(I_h(t)+I_u(t) \right) \\
E'(t)&=\beta S(t) \left(I_h(t)+I_u(t) \right) -\frac{1}{\tau} E(t) \\
I_h'(t)&=\rho\frac{1}{\tau}E(t) -\frac{1}{T_h} I_h(t) \label{B1}\tag{1} \\
I_u'(t)&=(1-\rho)\frac{1}{\tau} E(t) -\frac{1}{T_u} I_u(t) \\
H'(t)&=\frac{1}{T_h} I_h(t) 
 \end{align*}

The parameter $\beta$ represents transmission rate for infectious individuals, $\tau$ is mean incubation period, $\rho$ is fraction of hospitalized cases, $T_h$ is time from infectiousness (symptom) onset until hospitalization/isolation and reporting, and $T_u$ is the mean infectious period for an unhospitalized case.  Note that we consider an infectious case being hospitalized, isolated and reported as the same, whether or not the individual is treated in a hospital or placed in an isolation unit.  We assume that such a case is reported at the time of hospitalization and neglect any possibility of hospital transmission.  While transmission in a health-care setting poses significant risk to workers and visitors during Ebola outbreak, community transmission accounts for the vast majority of cases.  In addition, the main focus of our article is a detailed model of contact tracing, so we elect to make the model simpler by not including hospital transmission.  For this reason, we also neglect another route of transmission particular to Ebola; post-death transmission due to improperly handled deceased \cite{Weitz}, and leave incorporation of these features into contact tracing models for future work.  

 \begin{figure}       
 \subfigure[][]{\label{1a} \includegraphics[width=.5\textwidth,height=.35\textwidth]{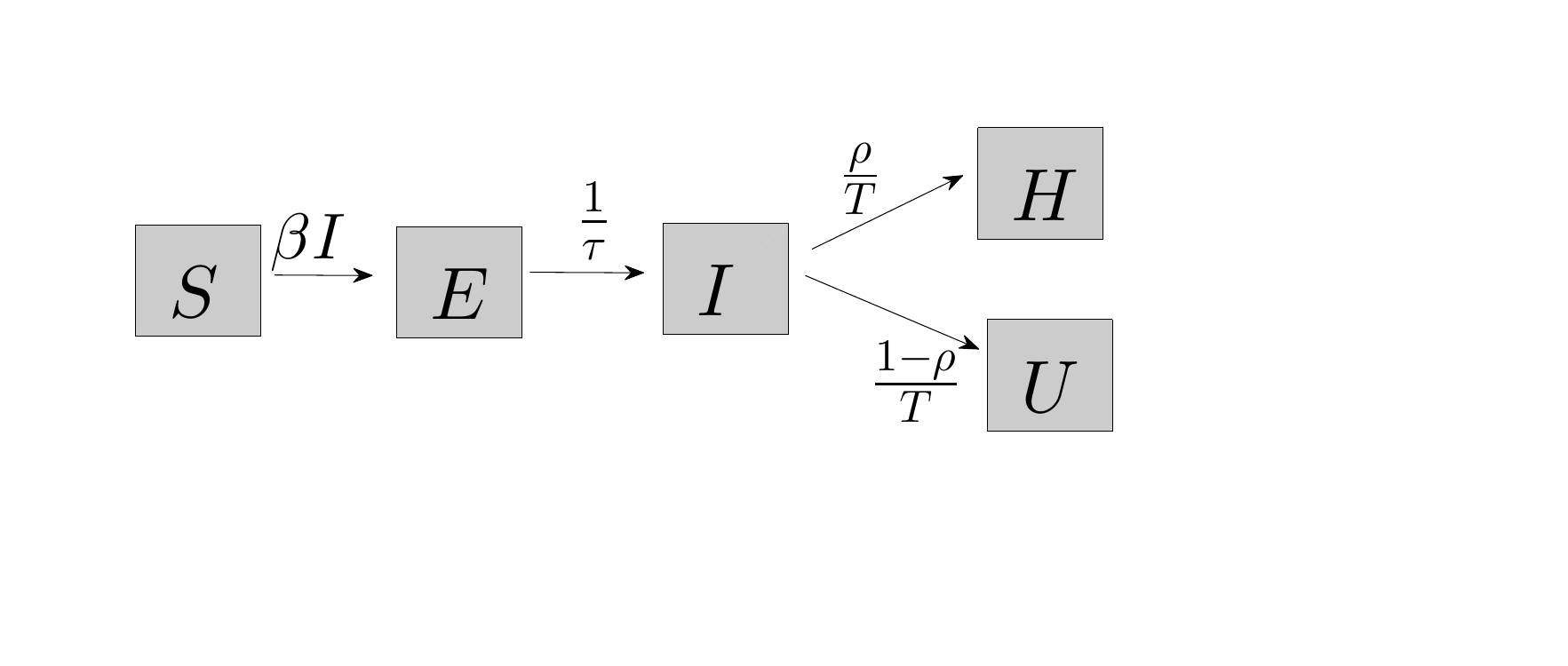}}
\subfigure[][]{ \label{1b}\includegraphics[width=.4\textwidth,height=.25\textwidth]{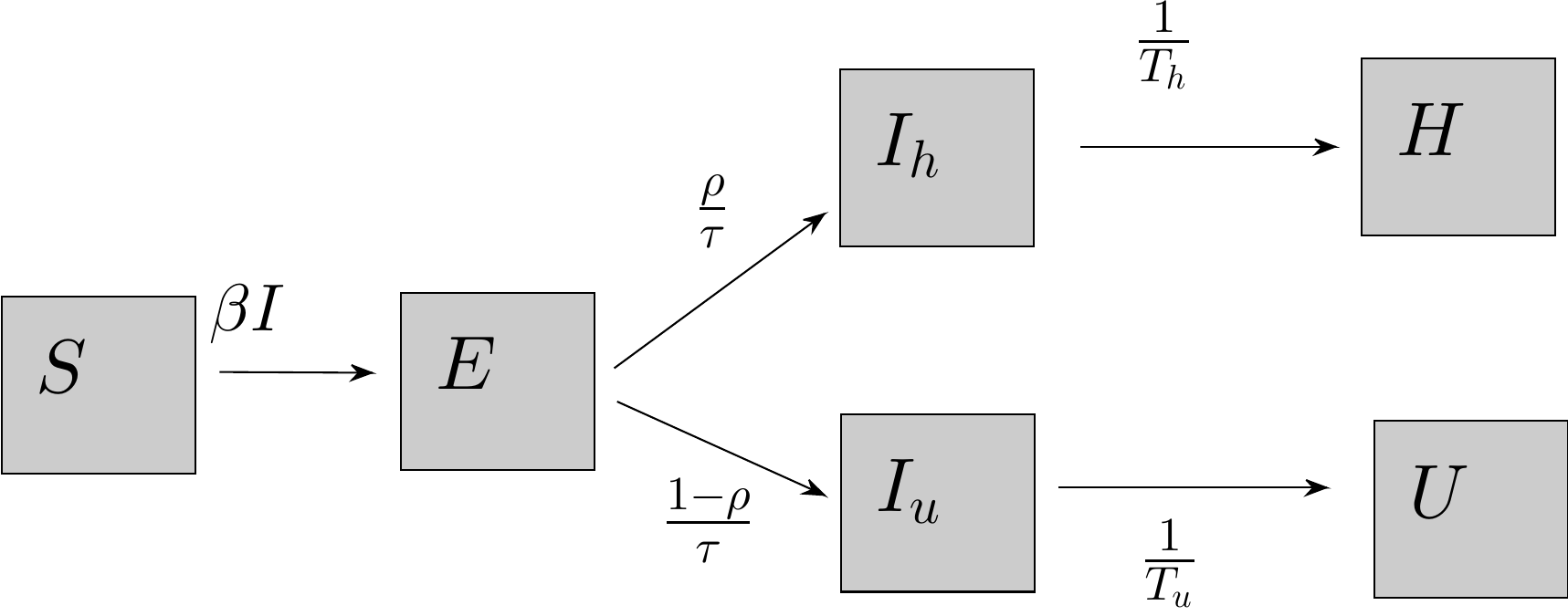}}
 \caption{ Base Models: (a)   Single infectious compartment model (equation (\ref{B2}) in \ref{A2}) with infectious period $T$ and proportion hospitalized $\rho$. (b)   Multiple infectious compartments model with variables for those which will be hospitalized, $I_h$, and unhospitalized, $I_u$, corresponding infectious periods $T_h$ and $T_u$, and proportion hospitalized $\rho$ (equation (\ref{B1})).  Informally, in (a), the ``hospitalization fate'' is determined upon exit from infectious class, whereas in (b), the ``hospitalization fate'' is determined upon entry to infectious class.  Note that the model in (a) is a special case of the model in (b).}
 \end{figure}   

The distinction between reported cases and unreported cases is important when considering contact tracing since only reported cases can trigger contact tracing.  We explicitly separate infectious individuals into distinct compartments based on whether they will be reported/hospitalized or unreported.  Note that the distinct infectious compartments of system (\ref{B1}) are in a similar spirit to the model in \cite{lewnard2014dynamics}, where the authors divide the infectious class into individuals who will recover and those who will die.  An alternative base model, the classic $SEIR$ model shown in Figure \ref{1a}, has a single infectious compartment $I$ with mean infectious period $T$ and, again, fraction of hospitalized cases $\rho$.  It is proved in the \ref{A2} that such a model is actually a special case of (\ref{B1}); namely the case where $T_h=T_u= T$ and the initial conditions $I_h(0)=\rho I(0), I_u(0)=(1-\rho)I(0)$.   Thus, the multi-infectious compartment model (\ref{B1}) allows for more generality than the single infectious compartment counterpart, with an additional independent parameter.  Also, it may be unreasonable to assume that infectious time periods $T_h$ and $T_u$ are equal.  Therefore, we utilize the more general base model (\ref{B1}) with multiple infectious compartments, but note that our contact tracing formulation can be incorporated into either base model.

For model (\ref{B1}), we can consider the time until hospitalization, $T_h$, and the fraction of (untraced) hospitalized cases, $\rho$, as parameters which can be together targeted by public health measures (independent of active contact tracing).  The infectious period of unhospitalized cases, $T_u$, can be considered a property of the disease.  In particular, we can take  $T_u=\psi T_d + (1-\psi)T_r$ where $\psi$ is mortality rate (outside of hospital), and $T_d$ and $T_r$ are the times until death and recovery, respectively.  Note that instead of explicitly separating infectious cases that will be unreported into compartments representing death and recovery, we use a single unreported infectious compartment for simplicity and since $T_d\approx T_r$ (see \ref{A2} for justification).    The effective reproduction number, $\mathcal R_e$, for model (\ref{B1}) is given by 
\begin{align*}
\mathcal R_e=\beta S_0 \left(\rho T_h + (1-\rho)T_u \right) = \mathcal R_0 \left( \rho \frac{T_h}{T_u} + 1-\rho \right),  \tag{2} \label{baseRe}
\end{align*}
where $S_0$ is the initial susceptible population size and $\mathcal R_0=\beta S_0T_u$ is defined as the basic reproduction number (without hospitalization of cases).  Standard incidence \cite{de1995does} is often utilized for Ebola models instead of the mass-action transmission assumed here, in which case $\mathcal R_0$ is independent of $S_0$.  We remark that replacing the mass-action (per-capita) contact rate $\beta S$ with the standard incidence rate $\beta\frac{S}{N}$ does not affect the results in this article.  Also, the population size $N$ remains approximately constant over the relevant timescale of the outbreak, thus we can think of $N \approx S_0$ as being absorbed into the parameter $\beta$.

\subsection{Incorporating contact tracing into dynamic models of EVD}

We incorporate contact tracing into the SEIR model framework by deriving a mechanistic link between the retrospective action of contact tracers with forward disease dynamics.  First, we describe the recommended contact tracing process for Ebola.  According to the WHO guidelines for contact tracing in Ebola outbreaks \cite{WHOtracing}, contact tracing consists of three basic elements, namely, contact identification, contact listing and contact follow-up.  The process of contact identification begins upon an infected individual (alive or dead) being reported and isolated.  A public health officer then conducts an epidemiological investigation to determine the contacts of the reported case.  In the next step (contact listing), a team of contact tracers tracks the contacts identified from the first step.   If they are showing symptoms at the time of tracking, then the tracers will call for an ambulance team to remove and isolate the individual.   Otherwise, the tracked contacts are informed of their contact status and told to report any symptoms immediately to a provided telephone number, along with restricting contact with other people.  A surveillance team follows up on the tracked contacts by visiting them regularly during a 21-day (``approximately'' maximum incubation period) window, and calling the ambulance team if symptoms are present.  

In order to model this process, we first make the following assumptions on top of the base model:   
\begin{itemize}
\item The hospitalized/reported cases can trigger contact tracing.
\item Contacts of a reported case are traced with probability $\phi$ and there is a fixed delay $\delta$ between index case reporting and contact tracking.  If a traced contact is infectious at this time, they are immediately removed/isolated; otherwise they follow a monitoring protocol aimed at early detection and isolation if symptoms develop.  
  \end{itemize}
  Note that the probability $1-\phi$ accounts for untraced contacts in terms of, both, reported cases that do not trigger tracing and ``missed'' contacts from reported cases that trigger contact tracing, i.e. $\phi=\mu\xi$ where $\mu$ is fraction of reported cases that trigger contact tracing and $\xi$ is the average fraction of total contacts from a reported triggering case that are traced.  Another way of expressing $\phi$ is the following:
$$\phi=\frac{\text{total } \# \ \text{of traced contacts of reported cases}}{\text{total } \# \ \text{of contacts of reported cases}}.$$

   We now formulate the contact tracing terms for incorporation into the base model \ref{B1}.  Let $C_e(t), \ \left(C_i(t)\right)$ be defined as the cumulative number at time $t$ of traced infected individuals who are in exposed (infectious) class when first tracked.  Then
 \begin{align*}
C_e'(t+\delta) &= \hspace{-0.5cm} \underbrace{\frac{1}{T_h} I_h(t)}_{\substack{\text{rate of reporting}\\ \text{of infectious individuals}}} \int\limits_0^{\infty} \underbrace{\frac{1}{T_h}e^{-s/T_h}}_{\substack{\text{reported case}\\ \text{infectious period} \\  \text{probability density}}} \,   \underbrace{\int\limits_{t-s}^t }_{\substack{\text{integral over}\\ \text{ infectious period of } \\  \text{length $s$}}} \,  \underbrace{F_E(t-r+\delta)}_{\substack{\text{probability still}\\ \text{incubating $t-r+\delta$ } \\ \text{units after infection}}}  \underbrace{\phi \beta S(r) \, d r}_{\substack{\text{traced transmissions }\\ \text{infected at}\\ \text{time} \  r\in[t-s,t]}} \,ds, \\ 
C_i'(t+\delta)&=\frac{1}{T_h} I_h(t)\int\limits_0^{\infty} \frac{1}{T_h}e^{-s/T_h} \int\limits_{t-s}^t  \, (1-F_E(t-r+\delta)) \phi \beta S( r) \, d r \, ds, \label{cumCT} \tag{3}
\end{align*}
where the incubation period cumulative distribution $F_E$ can be more general than the assumed exponential distribution.  Since the average infectious period is relatively short, we can assume $S$ is approximately constant over the infectious period, i.e. 
$ \int_0^{\infty} \frac{e^{-s/T_h}}{T_h} \int_{t-s}^t  S( r)  d r  ds \approx T_h S(t).  $
 Then the total number of traced transmissions at time $t+\delta$, $C(t+\delta)$, satisfies
 \begin{align*}
 C'(t+\delta)=C_e'(t+\delta)+C_i'(t+\delta) & \approx \phi \beta S(t) I_h(t)
 \end{align*}

Motivated by these calculations, we define the following additional four compartments:   $E_c^e, \ (E_c^i) $ are exposed individuals who will be traced and incubating (infectious) when tracked; $I_c^e$ are infectious individuals who have been traced during incubation stage and put under monitoring protocol; $I_c^i$ are infectious individuals who will be traced/removed (since symptomatic when first tracked).  The modeling framework for incorporating tracing is to separate the fraction $\phi$ of \emph{reported case} contacts that are traced at the moment of transmission, thus establishing an explicit link between tracing and reported case transmission.  Since the action of tracers upon tracking a contact depends upon whether the traced contact is displaying symptoms, we further separate transmissions which will be traced into those that will be incubating upon tracking ($E_c^e$) and those who will be infectious upon tracking ($E_c^i$).   The probability $p_e$ of a contact traced individual being incubating when tracked should generally depend on the index case infectious period $T_h$, the delay between index case reporting and tracing $\delta$, and the incubation period of the traced case $\tau$.  Indeed $p_e$ can be explicitly calculated from (\ref{cumCT}) as:
\begin{align*}
p_e&=\frac{C_e'(t)}{C_e'(t)+C_i'(t)}= \frac{\tau}{ \tau +T_h}e^{-\delta/\tau}, \tag{4} \label{pe}
\end{align*}  
when $F_E(x)=e^{-x/\tau}$ (exponentially distributed incubation period) and assuming that $S(t)$ is constant over an individual case's infectious period.   For the fraction $(1-p_e)$ of traced individuals who are infectious when tracked, their \emph{average infectious period (prior to being tracked)}, $T_i$, can be approximated by considering the expected value of the time after symptom onset when a traced contact is tracked (given that they have progressed past their incubation period upon tracking), $\hat{T}_i$:
\begin{align*}
\frac{1}{T_i}= \frac{1}{\bar{T}} + \frac{1}{\hat{T}_i} , \ \ \text{where} \ \ \hat{T}_i& =  \frac{1}{T_h} \int_0^{T_h} \left[ T_h+\delta - \left(\int_{r}^{T_h+\delta} \frac{x e^{- (x-r)/\tau}}{\int_0^{T_h-r+\delta} e^{- s/\tau} \, ds} \, dx\right) \right] \, dr  \tag{5} \label{Ti} 
\end{align*}
with $\bar{T}=\rho T_h +(1-\rho) T_u$ is the (weighted) average infectious period independent of contact tracing.  The above formula (\ref{Ti}) is explained in more detail in the \ref{A1}.   
\begin{figure}
\includegraphics[width=.9\textwidth]{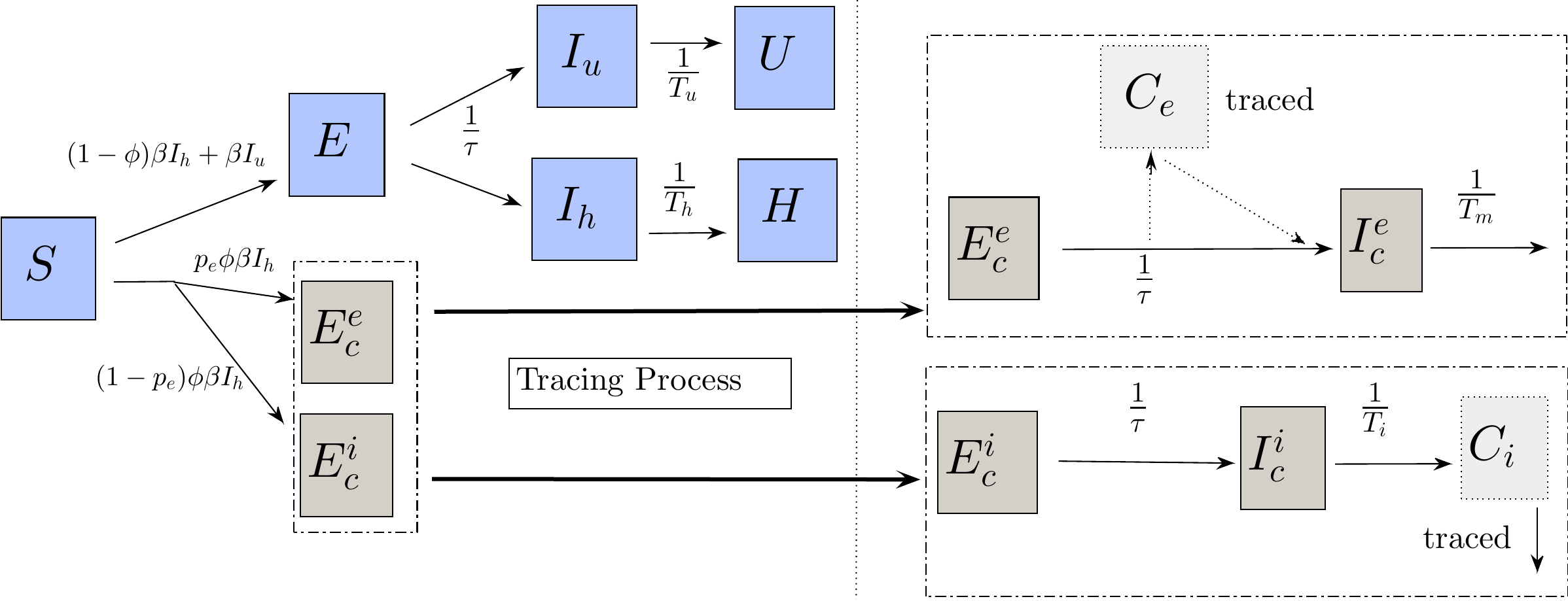}
\caption{\emph{Mechanistic model of contact tracing}:  Compartments dividing transmissions into those that are untraced, $E$ (whom will be unreported $I_u$ or reported $I_h$ and can trigger tracing), and those that are traced (whom will be tracked when incubating $E_c^e$ or infectious $E_c^i$).   The tracing process is depicted with the relevant sequence of events for the traced individuals.  Note that higher order tracing events are not illustrated in the diagram, but are included in the model (\ref{M2.1}).}
\end{figure}
We introduce the following additional assumptions and parameters:
\begin{itemize}
\item The infectious monitored individuals, $I_c^e$, have infectivity $\beta_m$ and mean infectious period (time from symptom onset until hospitalization) $T_m$.  The parameters $T_m$ and $\beta_m$ depend on the strictness and efficacy of the monitoring protocols.  Note that the infectious individuals traced during their infectious period, $I_c^i$, have infectivity $\beta$ prior to tracking and mean infectious period $T_i$.
\item Contact traced infectious individuals, $I_c^e$ and $I_c^i$, trigger contact tracing after being isolated/reported or tracked, which is termed \emph{higher order tracing}.  As with first-order tracing, the fraction $\phi$ of their contacts will be traced.  In particular, we assume that all $I_c^i$ individuals will be tracked, reported and trigger contact tracing, even if they die or recover before being traced. We introduce the parameters $p_e^e$ ($p_e^i$) representing the probability a higher order traced contact will be tracked during their incubation period after being \emph{infected by an $I_c^e$ ($I_c^i)$ individual}.  
\item The probabilities $p_e^e$ and $p_e^i$ can be calculated similar to $p_e$ (\ref{pe}), except using index case infectious periods $T_m$ and $\hat{T}_i$ respectively.   Note that the infectious period for a \emph{higher order tracing} infection caused by a $I_c^i$ individual should be smaller than $T_i$.  But this would require additional compartments, so we choose to approximate the infectious period as $T_i$ and ``close the system'' here. 
\end{itemize}
 
With all of these assumptions, we arrive at the following model:
\begin{align*}
S' \ \ & = \ \ -\beta S(I_h+I_u) - \beta S I_c^i - \beta_m S I_c^e \\
E'\ \ &= \ \ \beta SI_u+ \left(1-\phi \right) \beta S I_h +  (1-\phi) \beta S I_c^i +(1-\phi) \beta_m S I_c^e - \frac{1}{\tau} E \\
(E_c^e)'\ \ &= \ \ p_e\phi \beta S I_h + p_e^i \phi  \beta S I_c^i +p_e^e  \phi \beta_m SI_c^e - \frac{1}{\tau} E_c^e \\
(E_c^i)' \ \ &= \ \ \left(1-p_e\right)\phi \beta S I_h + \left(1-p_e^i \right) \phi  \beta I_c^i S+\left(1-p_e^e\right)  \phi  \beta_m S I_c^e - \frac{1}{\tau} E_c^i \\
I_h'  \ \ &= \ \ \frac{\rho}{\tau} E -\frac{1}{T_h}I_h   \label{M2.1} \tag{6} \\
I_u' \ \ &= \ \ \frac{1-\rho}{\tau} E -\frac{1}{T_u} I_u \\
(I_c^e)' \ \ &= \ \ \frac{1}{\tau} E_c^e -\frac{1}{T_m}I_c^e  \\
(I_c^i)' \ \ &= \ \ \frac{1}{\tau} E_c^i -\frac{1}{T_i}I_c^i 
\end{align*}

The model provides a detailed dynamic description of contact tracing, reconciling the \emph{forward dynamics} of the epidemic with the action of tracers working \emph{backward along a chain of contact}.  The variable and parameter descriptions are listed in Table \ref{Table1}.  Despite the number of compartments, several of the parameters are dependent on basic independent parameters.  Also, efficacy of contact tracing is determined by \emph{both} disease characteristics (through probability $p_e$ and $T_i$ as function of incubation period $\tau$ and infectious periods $T_h$, $T_u$) and monitoring/response efficacy (through delay $\delta$ in calculation of $p_e$, and parameters $\beta_m$ and $T_m$).  

There are limitations and corresponding generalizations of the model for future work.  For example, it is assumed that the incubation and infectious period are exponentially distributed, and infectivity (transmission) is uniformly distributed throughout the infectious period.  These are inherent assumptions of the ordinary differential equation system, however they may over-simplify the true epidemiological distributions.    Most likely these simplifications underestimate the impact of contact tracing.  Indeed, an exponentially distributed incubation period puts too much weight in small incubation times, which decreases the value of $p_e$.  A more appropriate distribution, for instance the gamma or Weibull distribution fit to Ebola patient incubation period data  \cite{Team, Chowell}, will result in higher vales of $p_e$, which increases the traced contacts that are monitored.  
In addition, assuming a uniform infectivity distribution, i.e. constant transmission rate, may not be consistent with the fact that viral load will steadily increase after symptom onset and peak infectivity will occur some time after the infectious period begins.  Thus, early removal of infectious individuals may drastically reduce transmission potential when these individuals are removed before peak infectivity \cite{Webb}.   Fraser et al. have considered contact tracing in an infection-age model with general infectivity distribution, however they did not explicitly link tracing and transmission \cite{Fraser}.  Infection-age structured or stage structured models can include the aforementioned disease features, along with post-death transmission, and incorporating tracing into these types of models is left for future work.  Despite the limitations, our ODE formulation offers an analytically tractable, mechanistic model of contact tracing in disease outbreaks.

\begin{table}[ht!]
\centering 
\begin{tabularx}{\textwidth}{>{} lX}
\toprule
Variable/Parameter & Description\\ [0.5ex]
\toprule
                  \textbf{Variables}       &                                     \\  \hline
     		 $S(t)$ & Susceptible individuals \\ \hline
		  $E(t)$ & Exposed (incubating) infected individuals \\ \hline
		  $I(t)$ & Infectious individuals \\ \hline
		   $E_c^e(t)$ & Exposed individuals who will be traced and incubating when tracked \\ \hline
		   $E_c^i(t)$ & Exposed individuals who will be traced and infectious when tracked \\ \hline
          $I_c^e(t)$ & Monitored infectious individuals (previously contact traced when incubating/not symptomatic).  \\ \hline
           $I_c^i(t)$ & Infectious individuals who will be traced and removed upon tracking (since symptomatic when first tracked).  \\ \hline
           \textbf{Independent parameters}       &                                     \\  \hline
		 $\beta$ & Transmission rate (for unmonitored infectious)   \\     \hline 
		 $ \tau$ & Mean incubation period \\ \hline
		 $ T_h$ & Time until hospitalization/reporting independent of contact tracing \\ \hline
		  $ T_u$ & Mean infectious period for unhospitalized case \\ \hline
		 $ \rho $ & Fraction of untraced cases that are isolated/reported (and can trigger tracing) \\ \hline
		 $ \phi $ & Probability that contact of reported case is traced \\ \hline
		 $ \delta $ & Time delay between case reporting and contact tracing \\ \hline
				$\beta_m$ & Transmission rate for monitored infectious ($I_c^e$) individuals   \\     \hline 
				  $ T_m $ & Mean infectious period for monitored infectious individuals (time between symptom onset and isolation)  \\ \hline
				\textbf{Dependent parameters}       &                                     \\  \hline
				 $ p_e $ & Probability of (first order) traced contact being incubating upon tracking (calculated by formula (\ref{pe}), depends on $\tau$, $T$, $\delta$)  \\ \hline
  $T_i$ & Mean infectious period (prior to being tracked) for contact traced infectious ($I_c^i$) individuals (calculated by formula (\ref{Ti}), depends on $\tau$, $T$, $\delta$)  \\ \hline
  $ p_e^e $ & Probability of (higher order) traced contact (infected by $I_c^e$ individual) being incubating upon tracking (calculated by formula (\ref{pe}) using $\tau$, $T_m$, $\delta$)  \\ \hline
  $ p_e^i $ & Probability of (higher order) traced contact (infected by $I_c^i$ individual) being incubating upon tracking (calculated by formula (\ref{pe}) using $\tau$, $T_i$, $\delta$)  \\ \hline
 
   \bottomrule
   \end{tabularx}
    \caption{Model variables and parameters.}\label{Table1}
   \end{table}

\section{Results}
\subsection{Formulas for $\mathcal R_e$:  control and probing properties of tracing} \label{RuleOfThumb}
In this section, we consider special cases of the general model (\ref{M2.1}) in which explicit formulas for $\mathcal R_e$ are simple and can be expressed completely in terms of contact tracing observables.  First, we consider the special case that $p_e \approx 1$ and $\beta_m=0$, i.e. long incubation period relative to infectious period, fast tracking and perfect monitoring. Although the formula (\ref{pe}) always yields $p_e<1$ for an exponentially distributed infectious period and incubation period, it is reasonable to assume that $p_e=1$ (traced individuals are always tracked during their incubation period) when the incubation period is large relative to infectious period.  In other words, we are assuming that contact tracing is ``perfect'' in the sense that traced individuals do not cause any secondary infections (because of effective monitoring and tracking).  Later in this section, we will relax this assumption and consider imperfect monitoring $\beta_m>0$ in the case of perfect reporting ($\rho=1$).

The effective reproduction number, $\mathcal R_e$, in the case of ``perfect'' tracing ($\beta_m=0$, $p_e=1$) simplifies to the following (calculation in \ref{ReAppendix}):  
\begin{align*}
\mathcal R_e= \mathcal R_0\left(\rho\frac{T_h}{T_u}(1-\phi) + 1-\rho \right), 
\end{align*}
where $\mathcal R_0:=\beta S_0T_u$ is the basic reproduction number, or more precisely the reproductive number in the absence of contact tracing and hospitalization, $\rho$ is the fraction of (untraced) cases which are reported and $\phi$ is the fraction of \emph{reported} case transmissions (contacts) which are traced.  Notice that in the special case of one mean infectious period for all cases, $T_h=T_u=T$,  the formula for $\mathcal R_e$ reduces to the following:
$$\mathcal R_e= \mathcal R_0(1-\rho\phi) $$
Thus, in this case, the contact tracing effort needed to reduce the effective reproduction number, $\mathcal R_e$, below one is: 
\begin{align*}
\mathcal R_e=1 & \Leftrightarrow \left(\rho\phi\right)_c=1-\frac{1}{\mathcal R_0},
\end{align*}
where $\left(\rho\phi\right)_c$ is the critical proportion of \emph{total} cases which need to be traced in order for $\mathcal R_e=1$.  This is the product of the fraction of cases which are reported, $\rho$, and the fraction of contacts (more precisely, transmissions) of reported cases which are traced, $\phi$.  This formula is similar to what was obtained in special cases of two branching process contact tracing models \cite{Muller,Klinkenberg}, but our approach differs from these previous works and is much simpler to derive.  Also, note the similarity with the well-known simple herd immunity vaccination proportion threshold $p_c=1-1/\mathcal R_0$ \cite{Vac}.

The goal now is to directly connect $\mathcal R_e$ with parameters that can be inferred by case and contact tracing data.  Considering the \emph{link between contact tracing and transmission}, we rewrite the fraction $\phi$ of transmissions from a  reported case which will be traced:
$$\phi=\frac{ \# \ \text{of traced contacts per reported case}}{\text{total } \# \ \text{of contacts per reported case}}:=\frac{\ell}{n}, $$
where these quantities are defined in a wholly susceptible population, i.e. the total population size is $N=S_0$.  Now the transmission rate $\beta$ can be written as $\beta=pc$, where $p$ is the probability of transmission per contact and $cS_0$ is the contact rate.  Now for a (untraced) reported case $n = cS_0T_h=\beta S_0T_h/p$. Define the parameter $k$ as the average number of secondary infected traced contacts identified per (untraced) reported case. Then $k:= p \ell = \beta S_0T_h/ \phi$.   Note that the value of the parameter $k$ can be estimated from data and records.  
 
 Another important parameter that can be obtained from case and tracing data is the fraction of reported cases which are traced contacts, denoted by $q$.  We find that $q$ can be expressed as the following constant in terms of model parameters:  
  \begin{align*}
q=\frac{  \phi}{\phi + (1-\phi)\rho  + (1-\rho) T_u/T_h }.  \label{qn} \tag{7}
\end{align*}
Indeed, one can interpret the denominator (multiplied by $\rho$) as the probability a case is reported considering the different potential routes to reporting:  traced; untraced, infection caused by reported case; untraced, infection caused by unreported case.  Also, notice in the case of 100\% reporting ($\rho=1$), $q$ reduces to $q=\phi$.  In addition, from a dynamic perspective, we define the fraction from the cumulative reported cases over some time window $[t-a,t]$ which are traced contacts:
 $$\tilde{q}_a(t):=\frac{\frac{1}{T_m}\int_{t-a}^t I_c(s)\,ds}{\frac{1}{T_m}\int_{t-a}^t I_c(s)\,ds+\frac{1}{T_h}\int_{t-a}^t I_h(s)\,ds}.$$
 In the \ref{ReAppendix}, it is proved that $\tilde{q}_a(t) \approx q$ for sufficiently large $t$, small $a$; and the convergence is quite rapid as illustrated in Figure \ref{Qconv} for the case $a=7 \ days$.  
 
 Combining the formula (\ref{qn}) and $k= \beta S_0 T_h \phi$,
 we find the following formula for the effective reproduction number $\mathcal R_e$:
 \begin{align*}
\mathcal R_e&= k\left( \frac{1-q}{q} \right)  \tag{8} \label{q1}
 \end{align*}
 Thus, $\mathcal R_e$ can be formulated simply as the product of the number of infected traced contacts per (untraced) reported case and the odds that a reported case is not a traced contact in the case of perfect tracing.  Notice that this formula is composed of two quantities that can be determined by case data and contact tracing records.

  \begin{figure}[t]       
 \subfigure[][]{ \label{Contour_kq} \includegraphics[width=.5\textwidth,height=.35\textwidth]{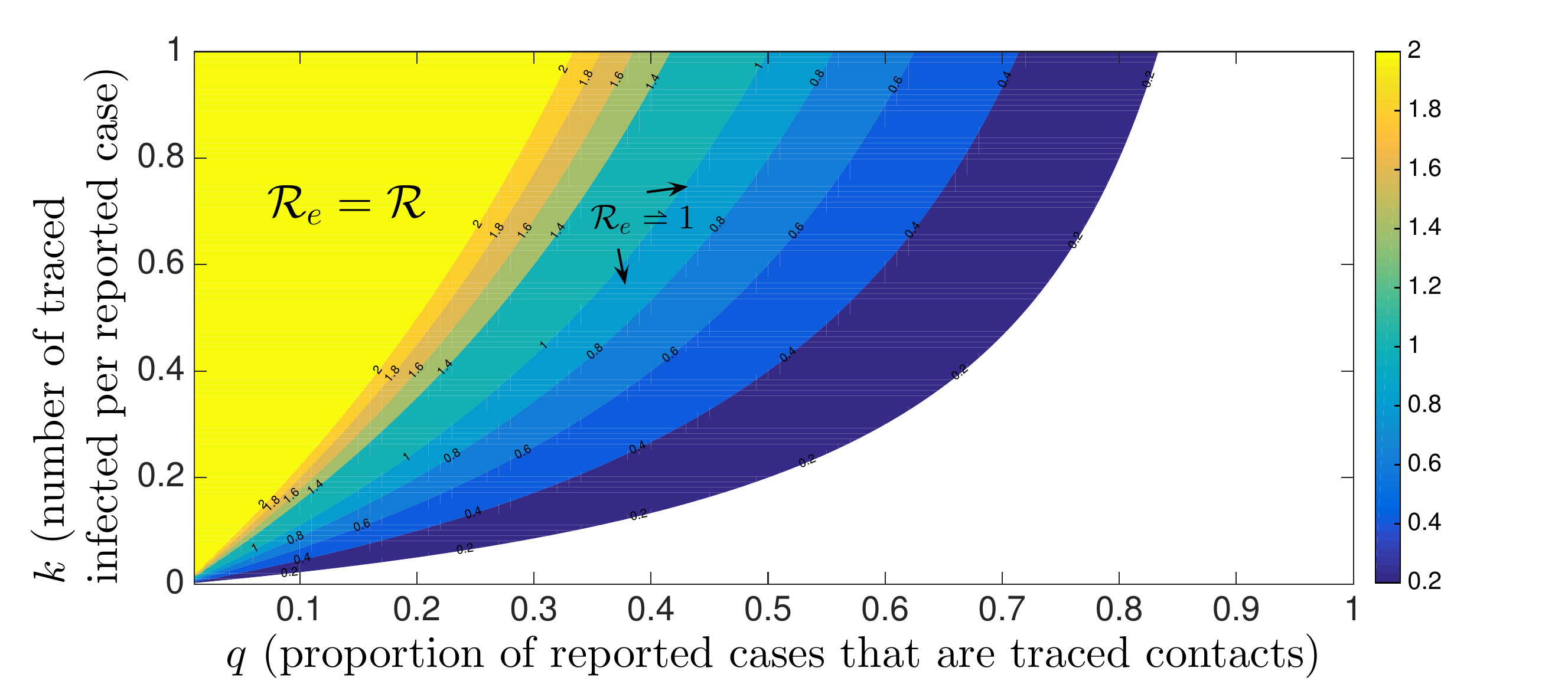}}
 \subfigure[][]{ \label{Qconv} \includegraphics[width=.5\textwidth,height=.35\textwidth]{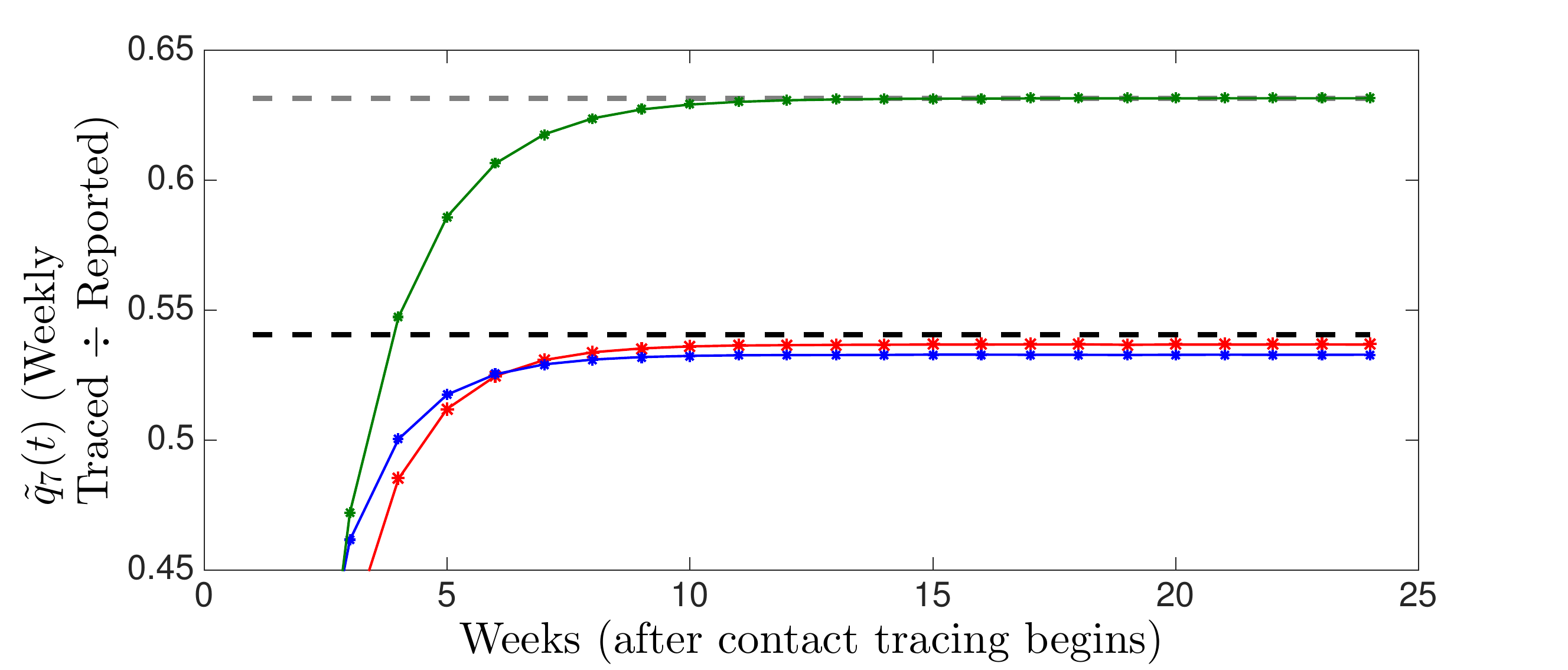}}
  \caption{ (a)  Contour map of $\mathcal R_e$ with respect to $k$ and $q$ from formula (\ref{q1}). (b) The weekly cumulative fraction of reported cases that are traced contacts, $\tilde{q}(t)$, converges to approximately (or to exactly) the quantity $q$ given in formula (\ref{qn}). }
  \end{figure}

 The intuition behind the formula for $\mathcal R_e$ is clear in the case where $100 \% $ of the cases are reported.  Indeed, consider a reported case which has caused $k+m$ secondary infections, where $k$ of these are traced (consistent with notation) and $m$ are untraced.  Note that this reported case is necessarily untraced since we are assuming that traced individuals are perfectly prevented from causing further infection.  If this reported case represents an average infectious individual in the population and the level of contact tracing here represents the average in the population, then the effective reproduction number $\mathcal R_e$ should be $m$ since this is the number secondary infections which can cause further transmissions.  Utilizing the formula (\ref{q1}) applied to this transmission chain, we correctly find 
 \begin{align*}
\mathcal R_e&=k\left( \frac{1-q}{q} \right)  =k\left( \frac{m/(k+m)}{k/(k+m)} \right) = m \label{simpCalc} \tag{9}
\end{align*}
 The somewhat surprising result is that formula (\ref{q1}) still holds with an arbitrary proportion of cases reported and distinct infectious periods for reported and unreported cases.  Thus contact tracing observables can theoretically provide an exact measure of $\mathcal R_e$ even in the presence of underreporting ($\rho<1$) and distinct infectious periods ($T_h\leq T_u$).  Our simple explanation utilizing the single index case and single level of transmission chain above will not suffice in the case of $\rho<1$ since we can add unreported cases which would not appear in the calculation  (\ref{simpCalc}), but should appear in $\mathcal R_e$.  However, a simple intuitive line of reasoning is that underreporting will cause cases to not be counted in individual transmission chains, but will also add to the total reported cases since additional reported cases that are transmissions of unreported cases will appear, thus balancing out the effect of underreporting.  A more complete heuristic argument of this result with transmission chains probably requires multiple levels and multiple index cases of an outbreak transmission tree.  Of course, considering explicit transmission trees more in depth leads to thinking about contact tracing on networks, which is a subject in itself.  We leave this as ongoing and future research, and provide further discussion about extending our work to networks in Section \ref{discussion}.

 In order to quantify the impact of control measures on reduction in reproduction number, we can calculate the relative change in reproduction number after applying the intervention.  Define $\mathcal R_e^o$ as the effective reproduction number when there is only case reporting/hospitalization independent of contact tracing (e.g. passive case finding \cite{whitty2014}) and no contact tracing.  Thus $\mathcal R_e^o= \mathcal R_0\left(\rho\frac{T_h}{T_u} + 1-\rho \right)$.  The relative changes in reproduction number with respect to the distinct interventions, passive case finding and contact tracing, are as follows:
\begin{align*}
\text{relative change in } \mathcal R_e \ \text{with tracing} \ &=\frac{\mathcal R_e^o-\mathcal R_e}{\mathcal R_e^o}= \frac{\rho q}{1-q(1-\rho)}, \\
\text{relative change in } \mathcal R_0 \ \text{with only passive case finding} \ &=\frac{\mathcal R_0-\mathcal R_e^o}{\mathcal R_0}= \rho\left(1-\frac{T_h}{T_u}\right), \tag{10} \label{RelChange} \\
 \text{relative change in } \mathcal R_0 \ \text{with passive case finding \& tracing} \ &=\frac{\mathcal R_0-\mathcal R_e}{\mathcal R_0}= \rho\left[\frac{1-\frac{T_h}{T_u}\left(1-q\right)}{1-q(1-\rho)}\right]
\end{align*}
With these formulas, we can quantify the impact of contact tracing and compare with the strategy of only passive case finding.  Notice that the particular role of contact tracing in reducing the effective reproduction number increases with $q$.  Thus $q$ provides a simple observable measure of the impact of tracing, however the relative change in $\mathcal R_e$ also depends on fraction of cases reported independent of tracing, $\rho$.  

Next, we consider another special case, in particular the assumption of perfect tracing is relaxed.  Suppose that there is perfect reporting and tracking, but imperfect monitoring, i.e. $\rho=p_e=1$, $\beta_m>0$.  In this case of 100\% reporting, the reproduction number in the absence of contact tracing is simply $\beta T_h$, which we denote by $\mathcal R:=\beta S_0 T_h$.  In addition, define the reproduction number of contact traced (monitored) individuals as $\mathcal R_m=\beta_m S_0 T_m$.  Then $\theta=\mathcal R_m/\mathcal R$ represents the reduction in secondary cases of a traced individual compared to an untraced individual.  As in the previous case, $\mathcal R_e$ reduces to a simple formula:
\begin{align*}
\mathcal R_e=(1-\phi) \mathcal R +\phi \mathcal R_m
\end{align*}
Thus, in this instance, the critical proportion of \emph{total} cases which need to be traced in order to bring $\mathcal R_e$ below one is:
\begin{align*}
\mathcal R_e=1 & \Leftrightarrow \phi_c=(1-\theta) \left(1-\frac{1}{\mathcal R}\right).
\end{align*}

In order to describe $\mathcal R_e$ in terms of contact tracing observables for this case, define $k=\phi \mathcal R$ and $k_m=\phi \mathcal R_m$ to be the average number of traced infected secondary cases caused per primary reported untraced infected and traced infected, respectively.  In the case of 100\% reporting (and $p_e=1$), $\rho=1$, it is not hard to show that the fraction of reported cases which are traced contacts, $q$, is simply $q=\phi$.  Thus, we find that
 \begin{align*}
\mathcal R_e&= k\left( \frac{1-q}{q} \right)  + k_m \label{ReImperfect} \tag{12}
 \end{align*}
 The impact of tracing in this case can simply be quantified, since the effective reproduction number without tracing is $\mathcal R_e^o=\mathcal R $.  The relative change in $\mathcal R_e$ with tracing is: 
 \begin{align*}
\frac{\mathcal R_e^o-\mathcal R_e}{\mathcal R_e^o}= q(1-\theta) \label{RelChange2} \tag{13}
\end{align*}
Together formulas (\ref{ReImperfect}) and (\ref{RelChange2}) quantify both $\mathcal R_e$ and the particular impact of contact tracing on $\mathcal R_e$ for a given tracing/monitoring efficacy (in terms of secondary case reduction $\theta$ due to tracing) in the case of 100\% reporting.

\subsection{Weekly point estimates of $\mathcal R_e$ for Sierra Leone and Guinea}
The results in the previous section suggest that the effective reproduction number, $\mathcal R_e$, can be estimated directly from contact tracing and reported case data via simple formulas derived from special cases of our mechanistic model of contact tracing.  In this section, we will apply formulas (\ref{q1}) and (\ref{ReImperfect}) to recent contact tracing and reported case data from the Ebola outbreak in Sierra Leone and Guinea in order to obtain some weekly point estimates of $\mathcal R_e$ for these two countries under the assumption of the special cases discussed in the previous section, and in particular different tracing efficacies.  Then, we calculate the effective reproduction number without contact tracing, $\mathcal R_e^o$, in order to gauge the potential impact of contact tracing on the epidemic.  
First, we briefly compare some other methods of estimating the effective reproduction number in outbreaks from the literature.

A common method for estimating the reproduction number is to fit a deterministic model to incidence data in order to infer unknown parameters in the underlying model and then calculate the reproduction number from standard techniques (e.g. next-generation method) \cite{riley2003, bjornstad2002, cintron2009, Rivers, Pandey}.  Drawbacks to this approach are the lack of generality, context-specific assumptions,  potential identifiability issues along with the fact that parameters vary throughout the epidemic \cite{cori2013, Weitz}.  More recent statistical approaches based on the work of Wallinga and Teunis \cite{wallinga2004} utilize case incidence data and knowledge of the distribution of the serial interval to estimate time-varying values of $\mathcal R_e$ throughout an epidemic \cite{Cauchemez2006, cori2013}.  These statistical methods are more generic in the sense that the underlying model is an infection network, however their implementation can be complex and requires prior estimates of serial interval distribution.  

We propose a novel method for estimating up to date values of $\mathcal R_e$ based on contact tracing and reported case observables by utilizing formulas (\ref{q1}) or (\ref{ReImperfect}).  In particular, with reported case data and contact tracing records detailing, both, the traced infected individuals and their infector, we can directly utilize these formulas for calculation of $\mathcal R_e$ (which hold under the different simplifying assumptions of the previous section).  Indeed, suppose we have the reported case count at a particular time point $t$ (cumulative reported cases during time window $t-a$ to $t$), information about whether these cases are known traced contacts, along with records discerning the number of traced infected contacts which are derived from these reported cases.  Then the number of traced infected identified per untraced and traced reported case, $k$ and $k_m$, can be calculated at this time $t$.  Furthermore, the fraction of reported cases which are traced contacts, $q$, can be calculated for the subsequent time windows of reporting for these traced contacts derived from the infecting reported cases at time $t$.  Then, the formula (\ref{ReImperfect}), $\mathcal R_e=k(1-q)/q \ + k_m$, can be utilized with an appropriate averaging technique to estimate the effective reproduction number, $\mathcal R_e$, at time $t$.  Thus, with a time-series of this case/tracing data, the corresponding point estimates of $\mathcal R_e$ can be calculated in this way for each feasible time point in the series.  
 \begin{figure}[t]      
    \subfigure[][]{\label{Gdata} \includegraphics[width=.5\textwidth,height=.35\textwidth]{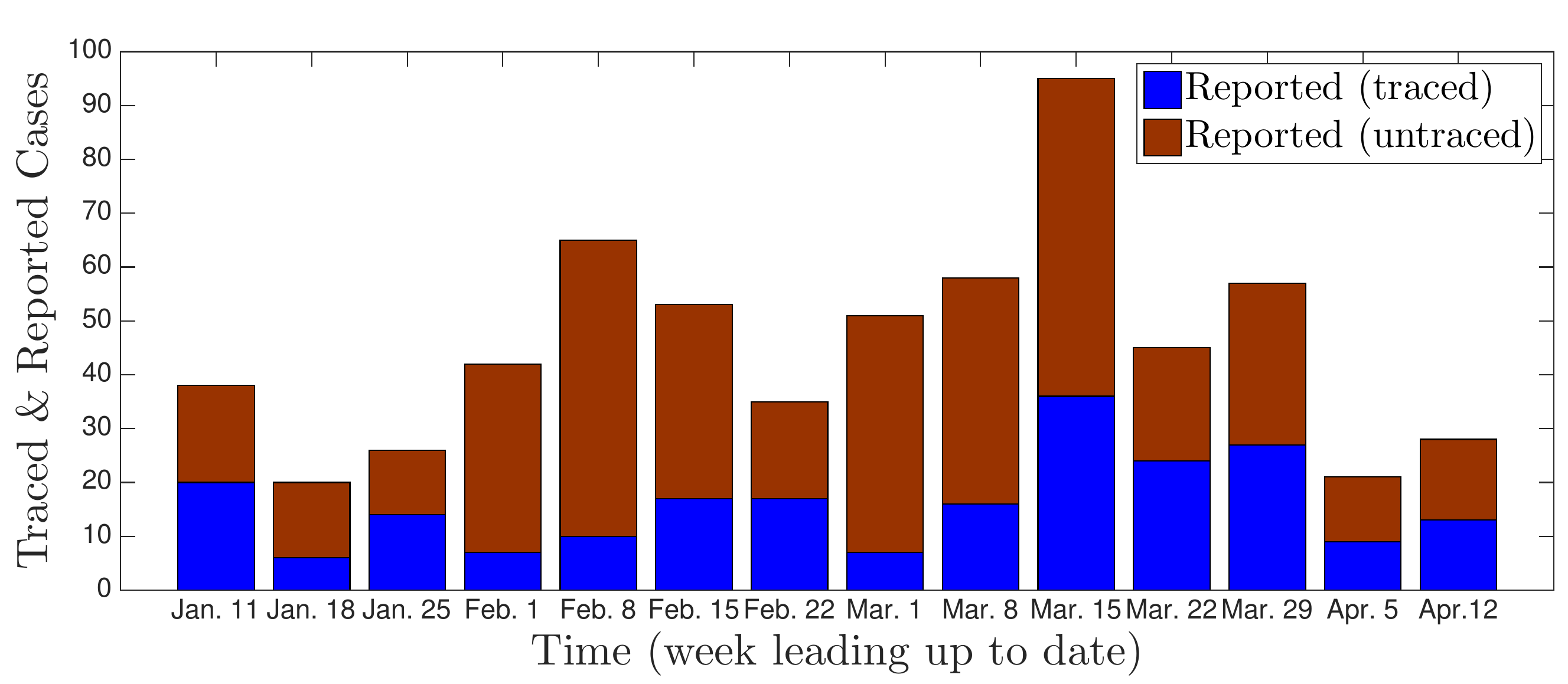}}
     \subfigure[][]{\label{Sdata} \includegraphics[width=.5\textwidth,height=.35\textwidth]{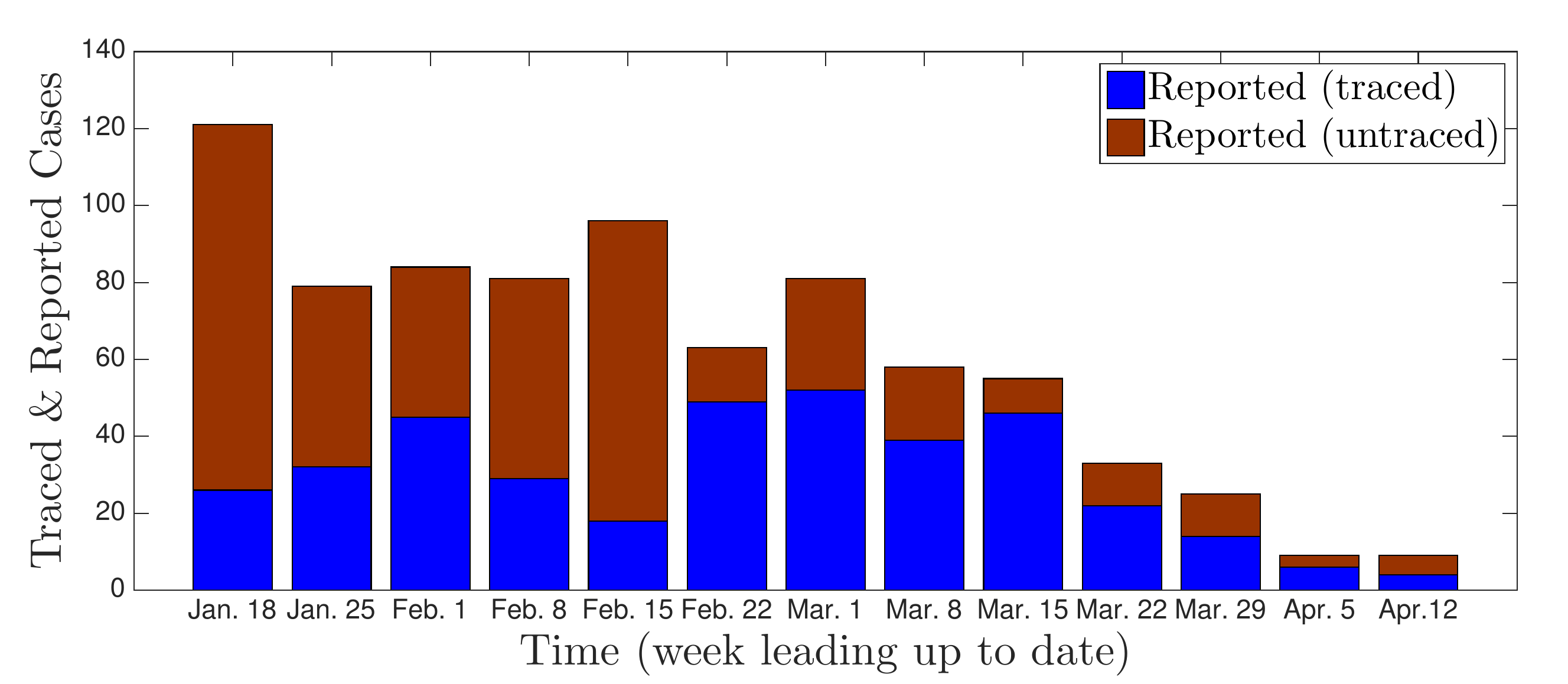}} \\
      \subfigure[][]{\label{GdataQ} \includegraphics[width=.5\textwidth,height=.35\textwidth]{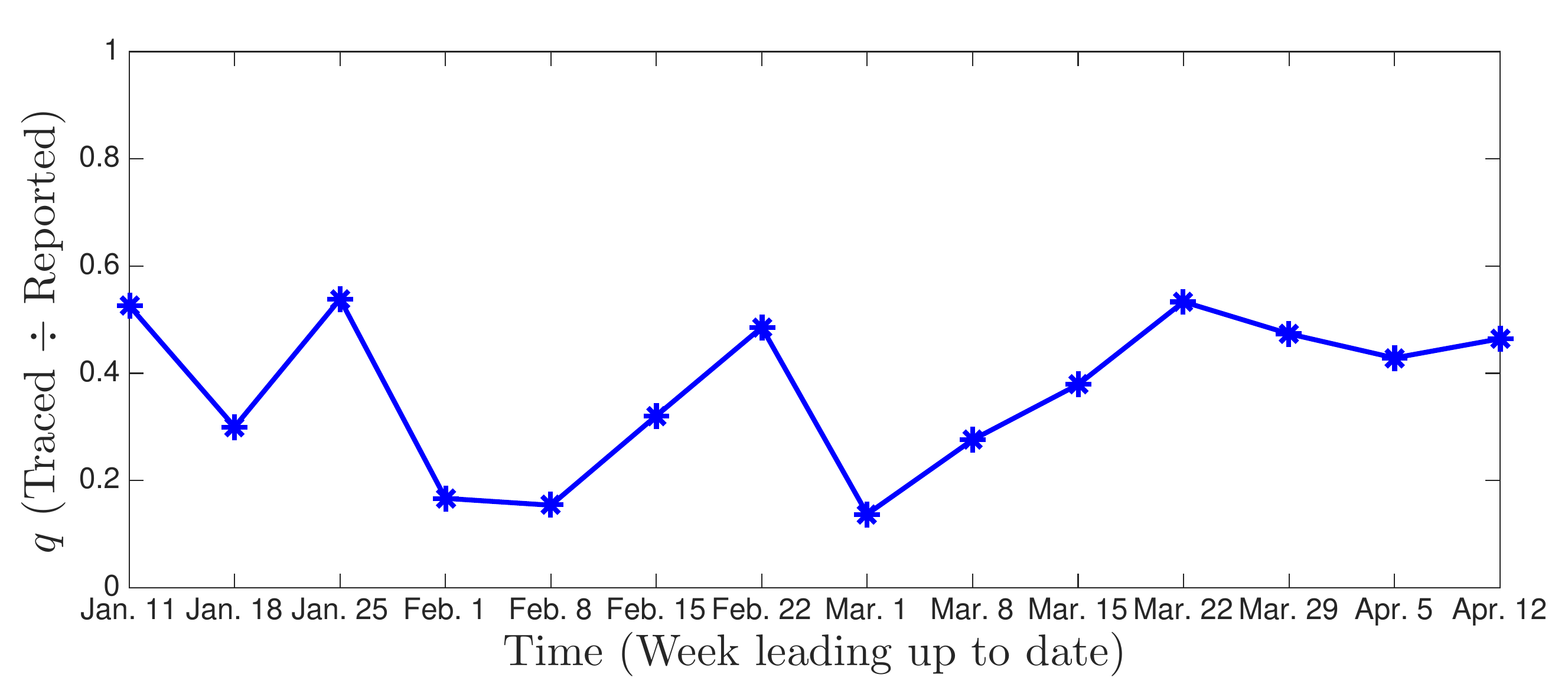}}
     \subfigure[][]{\label{SdataQ} \includegraphics[width=.5\textwidth,height=.35\textwidth]{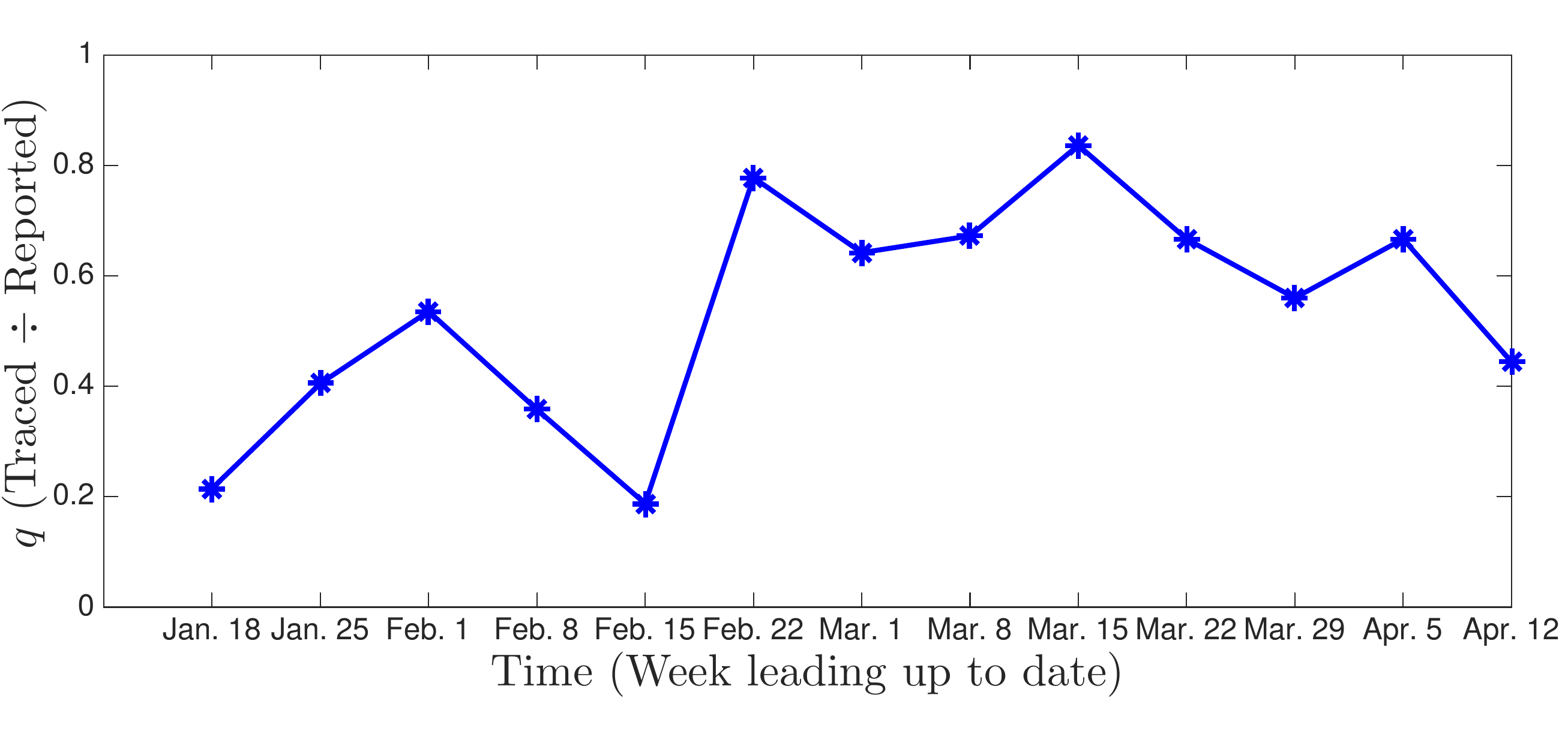}}
   \caption{WHO data of weekly reported and traced cases for Guinea (a) and Sierra Leone (b) for mid January to mid April.  The fraction of reported cases that are known traced contacts, $q$, is plotted versus week for Guinea (c) and Sierra Leone (d).  The data was obtained from weekly WHO situation reports (\cite{WHOSit}). }
   \label{WHOdataFig}
   \end{figure} 
   
In comparison with the other methods, our proposed procedure is similar to the statistical approach first described by Wallinga and Teunis \cite{wallinga2004}, however, our approach shows a potential way to, both, incorporate specific contact tracing data and remove the need for utilizing (and estimating) the serial interval distribution.  We remark that Cauchemez et al. utilize contact tracing data in order to infer the serial interval and use this to ``reconstruct tracing information for untraced cases'' assuming perfect reporting \cite{Cauchemez2006}.  Our work suggests that contact tracing and reported case data, by themselves, may be sufficient to estimate $\mathcal R_e$ based on formula (\ref{ReImperfect}).  Also, since our formulas are derived from a mechanistic differential equation model, we provide a bridge between compartmental deterministic models and the statistical methods based on reconstructing an infection network.  With our general detailed model of contact tracing in the background, we can explore the effect of relaxing simplifying assumptions about the tracing process.  In addition, the formulas (\ref{RelChange}) and (\ref{RelChange2}) provide a fresh perspective on  real-time evaluation of efficacy of control measures, in particular contact tracing, which was a problem considered in \cite{Cauchemez2006}. 

Recall that the formula was derived to hold with arbitrary underreporting ($\rho\leq 1$) in the case of perfect monitoring, $k_m=0$, and hold for 100\% reporting ($\rho=1$) in the case of imperfect monitoring, $k_m>0$.     It should be possible to obtain the $k$ and $k_m$ directly from contact tracing data, however, we were not able to find such data, and we recommend public health officials to keep and publicly share the necessary records for the estimation of average traced infected contacts found per reported case.   Despite this limitation of direct data, in a similar spirit to the methods of Wallinga and Teunis \cite{wallinga2004}, we utilize estimates of the serial interval distribution for Ebola from previous works \cite{Chowell}, in order to infer likely ``averaged'' infection networks, from which values of $k$ and $k_m$ are obtained under different scenarios of tracing efficacy.   In this way, we calculate weekly point estimates of $\mathcal R_e$ utilizing formula (\ref{ReImperfect}) with weekly case and tracing data, along with incorporating previous estimates of the serial interval distribution of Ebola, in order to evaluate the potential impact of contact tracing on the epidemic.

The WHO \cite{WHOSit} has recently published weekly counts of traced individuals among reported cases in Guinea and Sierra Leone as part of the weekly WHO situation reports (see Figure \ref{WHOdataFig}).  
 Recall, by (\ref{ReImperfect}), the effective reproduction number can be estimated as $\mathcal R_e=k\left( \frac{1-q}{q} \right)+k_m$, where $k$ and $k_m$ are the infected traced contacts identified per untraced and traced reported case, respectively.  Define $C_j$ and $T_j$ as the cumulative reported cases and traced (reported) cases in week $j$, respectively. The serial interval, which we denote by $\sigma$, is defined as the time from illness onset in the primary case to illness onset in the secondary case \cite{Fine}.  The serial interval of Ebola has been estimated to be a Weibull distribution (random variable denoted $W$) with scale and shape parameters estimated at 13.6 and 2.6, respectively \cite{Chowell}.  Therefore, we estimate the probabilities $p_1$, $p_2$, and $p_3$ of the serial interval being $\leq 7$ days, $7-14$ days, and $14-21$ days, as: $p_1=\mathbb P(W<7) =0.18$, $p_2=\mathbb P(7<W<14) = 0.51 $, and $p_3=1-p_2-p_1=0.31 \approx \mathbb P(14<W<21)$.  Suppose that the tracing/monitoring reproduction number reduction is $0\leq \theta\leq 1$, i.e. $\mathcal R_m=\theta \mathcal R$.  Then, the average number of traced infected identified in week $j+n$ per untraced and traced reported case from week $j$, $k^{j,n}$ and $k^{j,n}_m$, and the fraction of reported cases that are traced in week $j+n$, $q_{j+n}$, are:
 \begin{align*}
 k^{j,n}&=\frac{T_{j+n}}{C_j-(1-\theta)T_j}, \quad  k_m^{j,n}=\frac{\theta T_{j+n}}{C_j-(1-\theta)T_j}, \quad q_{j+n}=\frac{T_{j+n}}{C_{j+n}}, 
 \end{align*}
 since $k^{j,n} (C_j-T_j) + k_m^{j,n} T_j = T_{j+n}$ and $k_m^{j,n}=\theta k^{j,n}$.  
 
 Thus, we estimate the (expectation of) effective reproduction number of week $j$, $\mathbb E(\mathcal R_e^{(j)})$, as: 
  \begin{align*}
  \mathbb{E}(\mathcal R_e^{(j)})&=\sum\limits_n \mathbb{P}( n-j-1 \leq \sigma \leq  n-j \ \text{weeks}) \ \cdot \  \mathbb{E}(\mathcal R_e^{(j)} \ | \ n-j-1 \leq \sigma \leq  n-j \ \text{weeks}) \\
&= \sum\limits_n p_n \ \cdot \  \left(k^{j,n} \left(\frac{1- q_{j+n}}{q_{j+n}}\right) + k_m^{j,n} \right) \\
&= \frac{1}{C_j-(1-\theta)T_j} \displaystyle\sum\limits_n p_n \left( C_{j+n}-(1-\theta)T_{j+n} \right) \label{ReData} \tag{14}
 \end{align*}
 Therefore $\mathcal R_e^{(j)}$ is estimated as the ratio of the sum of untraced and a fraction $\theta$ of traced reported cases, weighted (based on serial interval distribution) over the weeks following $j$, to the sum of untraced and this fraction $\theta$ of traced reported cases during week $j$.  In the case of perfect tracing, $\theta=0$, only weekly untraced reported cases ($C_n-T_n$) are involved in the calculation of $\mathcal R_e$.  Conversely, in the case that tracing has no impact, $\theta=1$, all reported cases (traced or untraced) are equally weighted in the calculation of $\mathcal R_e$, thus the specific tracing data giving traced reported cases per week ($T_n$) is not utilized.   We note that a more rigorous maximum likelihood analysis can be developed with our approach similar to \cite{wallinga2004} and confidence intervals can be generated for the weekly estimates of $\mathcal R_e$.  However, since we are more interested in a relatively straightforward application of formula (\ref{ReImperfect}) to data and quantifying the potential impact of contact tracing, we simply calculate the expected value of the weekly reproductive number, $\mathcal R_e^{(j)}$, by the law of total expectation. 

  \begin{figure}[t]      
    \subfigure[][]{\label{GRe} \includegraphics[width=.5\textwidth,height=.35\textwidth]{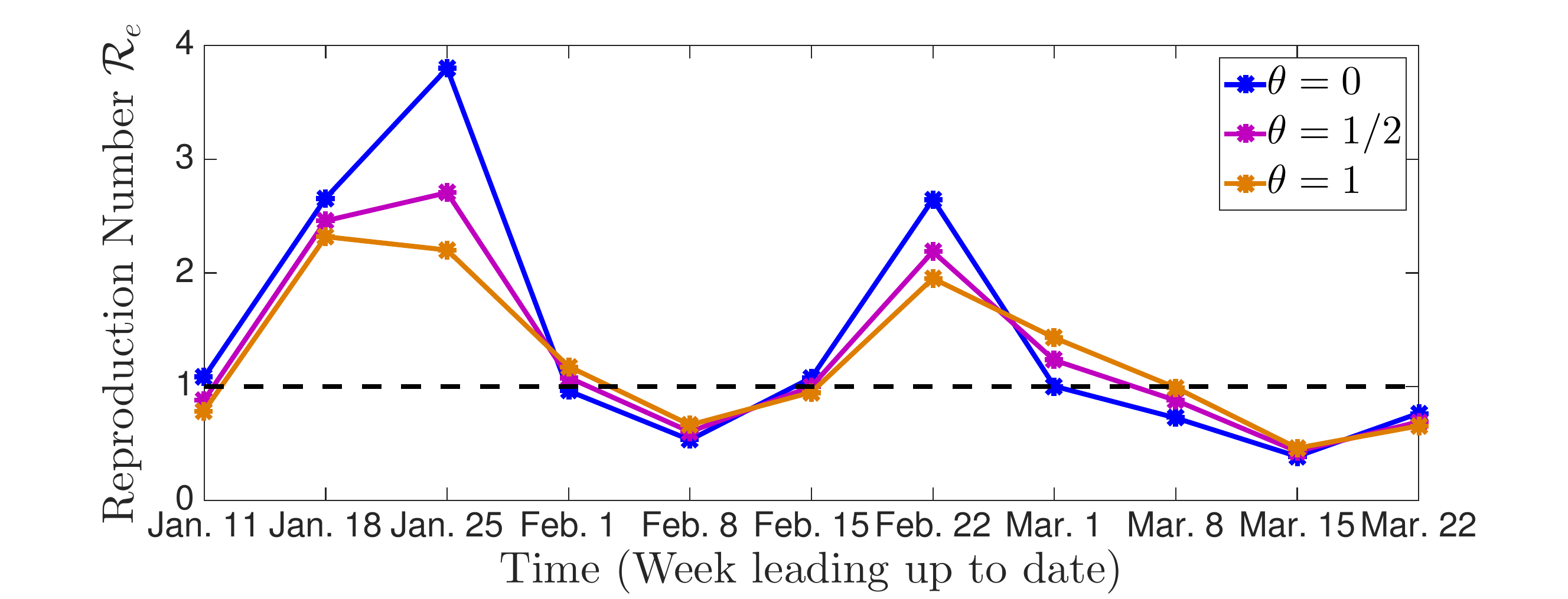}}
     \subfigure[][]{\label{SRe} \includegraphics[width=.5\textwidth,height=.35\textwidth]{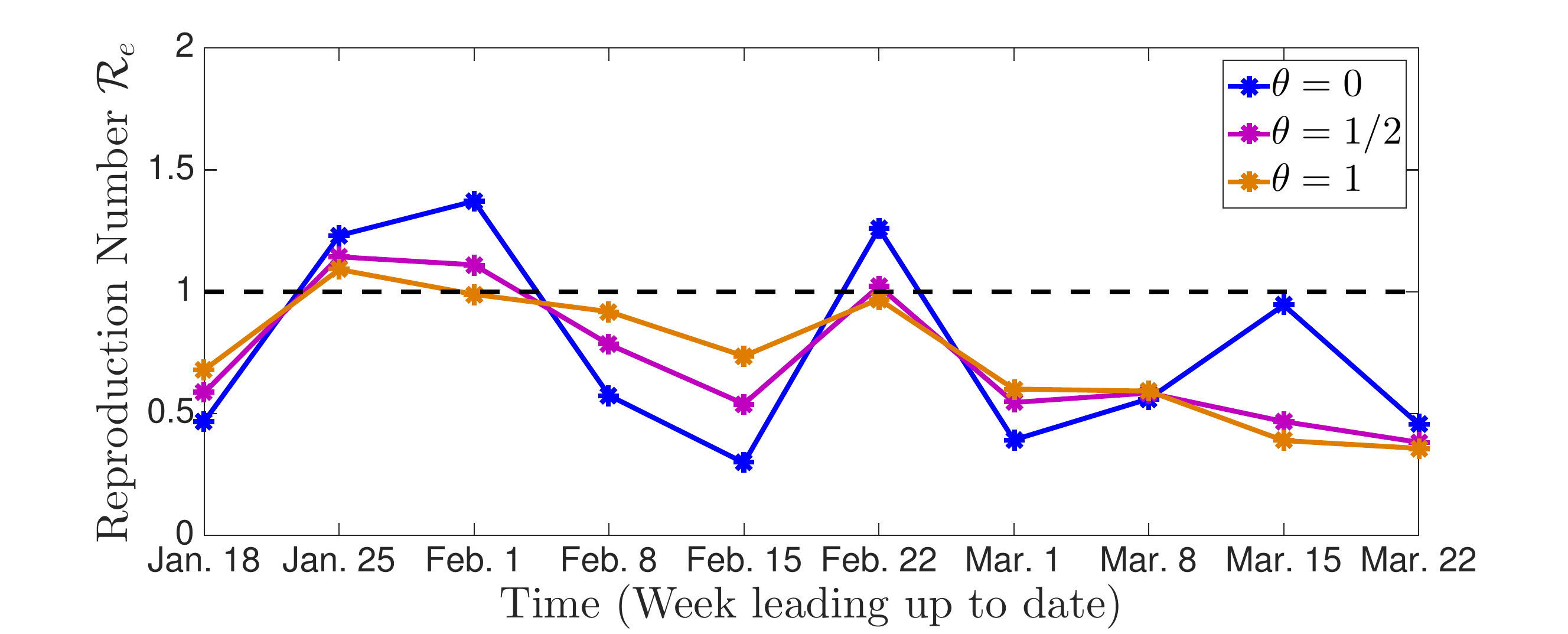}}  \\
      \subfigure[][]{\label{GReC} \includegraphics[width=.5\textwidth,height=.35\textwidth]{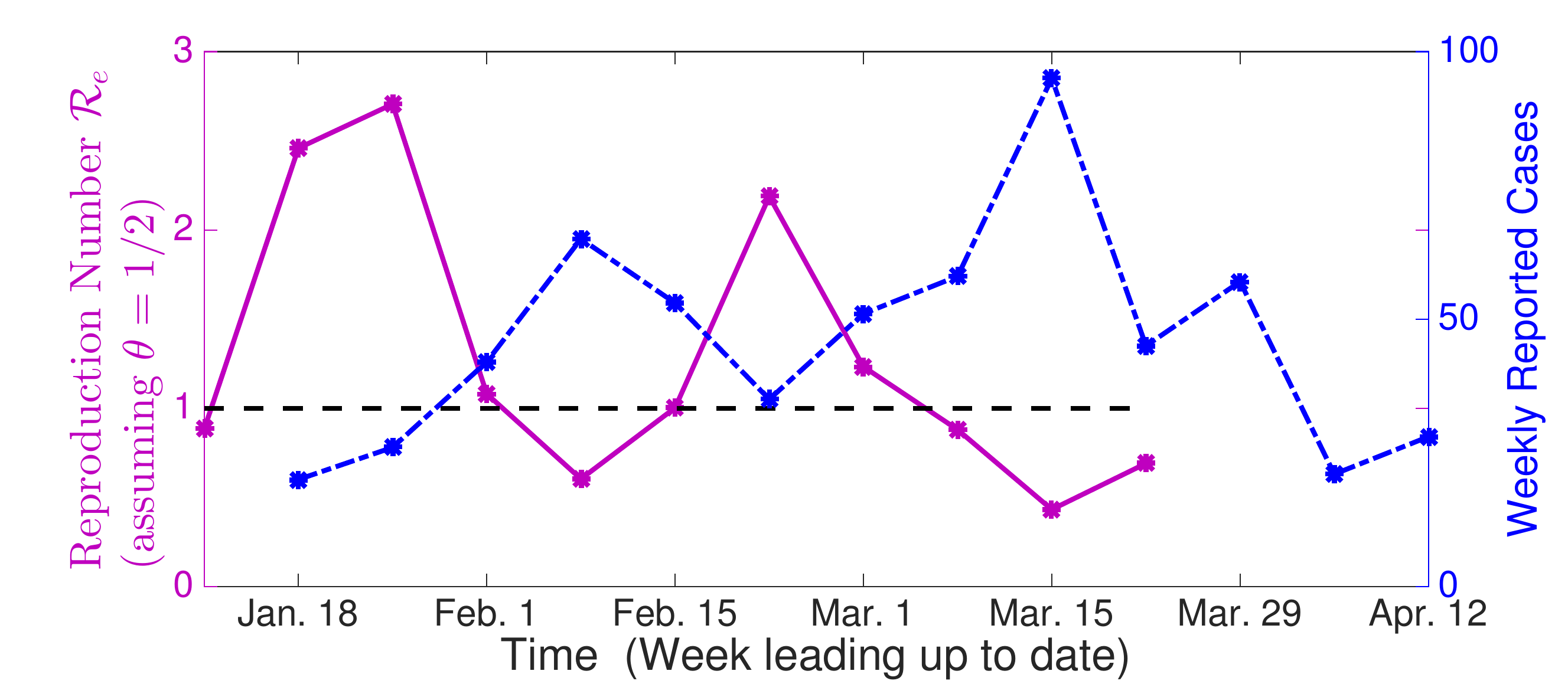}}
     \subfigure[][]{\label{SReC} \includegraphics[width=.5\textwidth,height=.35\textwidth]{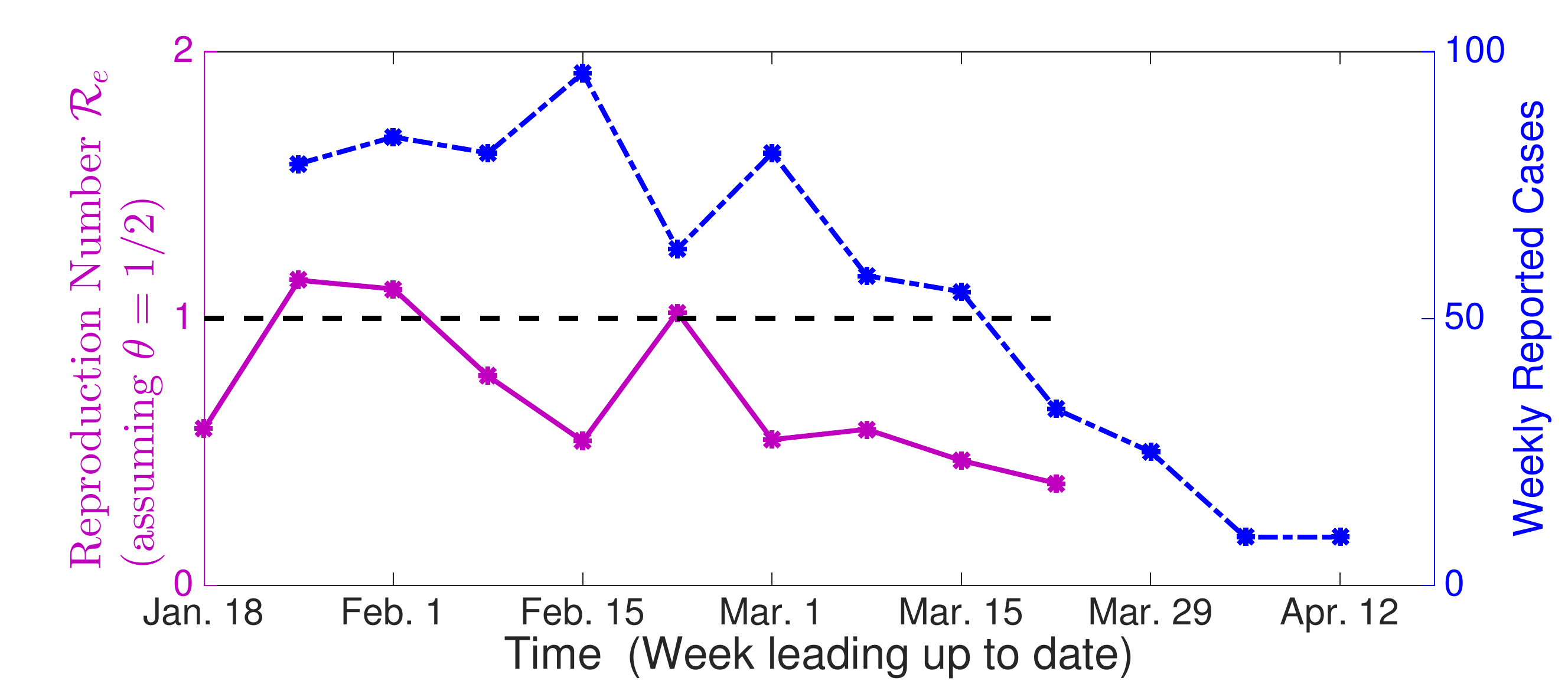}} 
      \caption{Weekly estimated values of  $\mathcal R_e$ ( $\mathbb{E}(\mathcal R_e^{(j)})$ in (\ref{ReData}) ) for both Guinea (a) and Sierra Leone (b) assuming tracing reproduction number proportion $\theta=0$, i.e. perfect tracing (blue), $\theta=0.5$ (purple), and $\theta=1$ (orange).  In (c) and (d), estimates of $\mathcal R_e$ assuming $\theta=0.5$ are shown alongside the weekly case count curve for Guinea and Sierra Leone, respectively. }
   \label{ReEst}
   \end{figure} 

In Figure \ref{GRe} and \ref{SRe}, the weekly point estimates of $\mathcal R_e$ for Guinea and Sierra Leone are obtained using the formula (\ref{ReData}) for different levels of tracing efficacy; reproduction number proportion $\theta=0$ (perfect tracing), $\theta=1/2$, $\theta=1$ (no impact of tracing).  Figure \ref{GRe} and \ref{SRe} shows the estimates of $\mathcal R_e$ in the case $\theta=1/2$ alongside the weekly reported cases for Guinea and Sierra Leone, respectively.  Notice that, especially for Guinea, the reproduction number $\mathcal R_e$ fluctuates substantially about the epidemic threshold $\mathcal R_e=1$ over the time period depicted.  These fluctuations coincide with changes in the weekly incidence curve with a delay representative of the serial interval.  For the Figures \ref{GRe} and \ref{SRe}, one should keep in mind that formula (\ref{ReData}) is utilized to calculate $\mathcal R_e$, thus reported cases \emph{and} the amount of traced cases affect the values of $\mathcal R_e$.  In comparison with Guinea, the estimates of $\mathcal R_e$ for Sierra Leone have a more decreasing trend with smaller values and less fluctuation, and this pattern also holds true in the reported case count.  In addition, Figures \ref{GdataQ}, \ref{SdataQ} show that the fraction of traced cases, $q$, is larger for Sierra Leone.  This suggests that contact tracing has been more successful in Sierra Leone over the time period, as will be discussed below and further depicted in Figure \ref{CTimpact}.  Supporting this notion, delving deeper into the given data in Figures \ref{GdataQ} and \ref{SdataQ}, we notice that the auto-correlation of the time-series of $q$ for Guinea is more random than that of Sierra Leone.  Indeed, for Sierra Leone the data for $q$ has larger auto-correlation corresponding to time lags less than three weeks, which is expected based on the serial interval of Ebola.  This suggests that perhaps Sierra Leone has had a more constant tracing effort when compared to Guinea, which may explain the different outcomes of the countries over this time period.  Further comparison of the countries' responses and epidemics is an important topic and will be studied in future work. 

 The reproduction number fluctuates more when assuming perfect tracing ($\theta=0$), as only the untraced cases are assumed to cause secondary infections and the fraction of reported cases that are traced contacts, $q$, is not fixed from week to week.  Certainly the limited amount of tracing data, combined with variability in $q$, may cause error in our calculation of $\mathcal R_e$.  In addition, recall the simplifying assumptions of either perfect tracing or perfect reporting under which formula (\ref{ReImperfect}) holds.  Also, it is unclear how much of the variation in $\mathcal R_e$ is due to observation error or underlying stochasticity in the transmission process versus actual changes in reproduction number (transmission potential).  The success of formula (\ref{ReImperfect}) in finding the ``true'' expected value of the (possibly time-varying) effective reproduction number $\mathcal R_e$ will be explored in future work with simulated data.  
Also, as mentioned previously, a major limitation is our implicit inference of the parameter $k$ and $k_m$ based on the estimated serial interval distribution.  Despite the limitations, our main goal in this section is to illustrate the potential value in utilizing contact tracing observables to estimate $\mathcal R_e$ through formula (\ref{ReImperfect}).  We emphasize that our general approach for calculating $\mathcal R_e$ with can avoid the need to utilize a serial interval distribution.  Instead, contact tracing observables can be directly used to calculate $k$ and $k_m$ in the formula (\ref{ReImperfect}) if the detailed tracing records are available.  

 \begin{figure}[t!]      
    \subfigure[][]{\label{G05} \includegraphics[width=.5\textwidth,height=.35\textwidth]{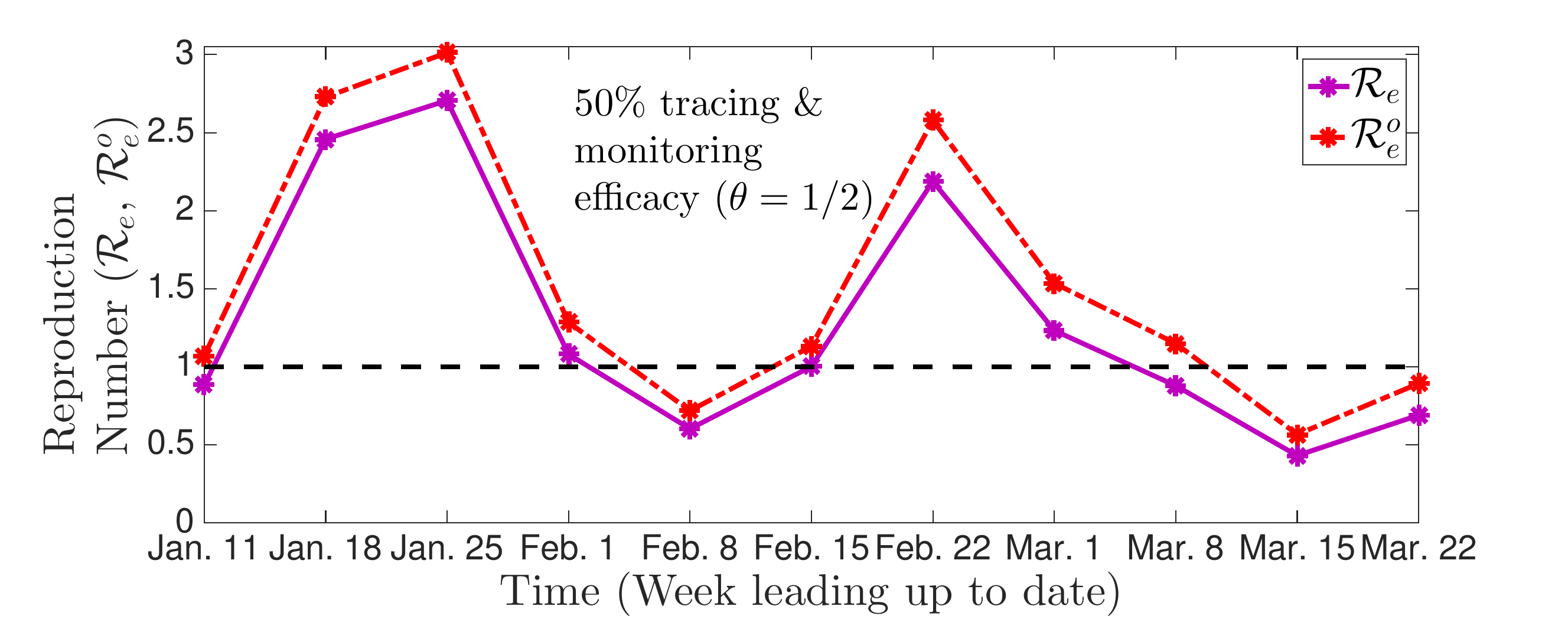}}
     \subfigure[][]{\label{S05} \includegraphics[width=.5\textwidth,height=.35\textwidth]{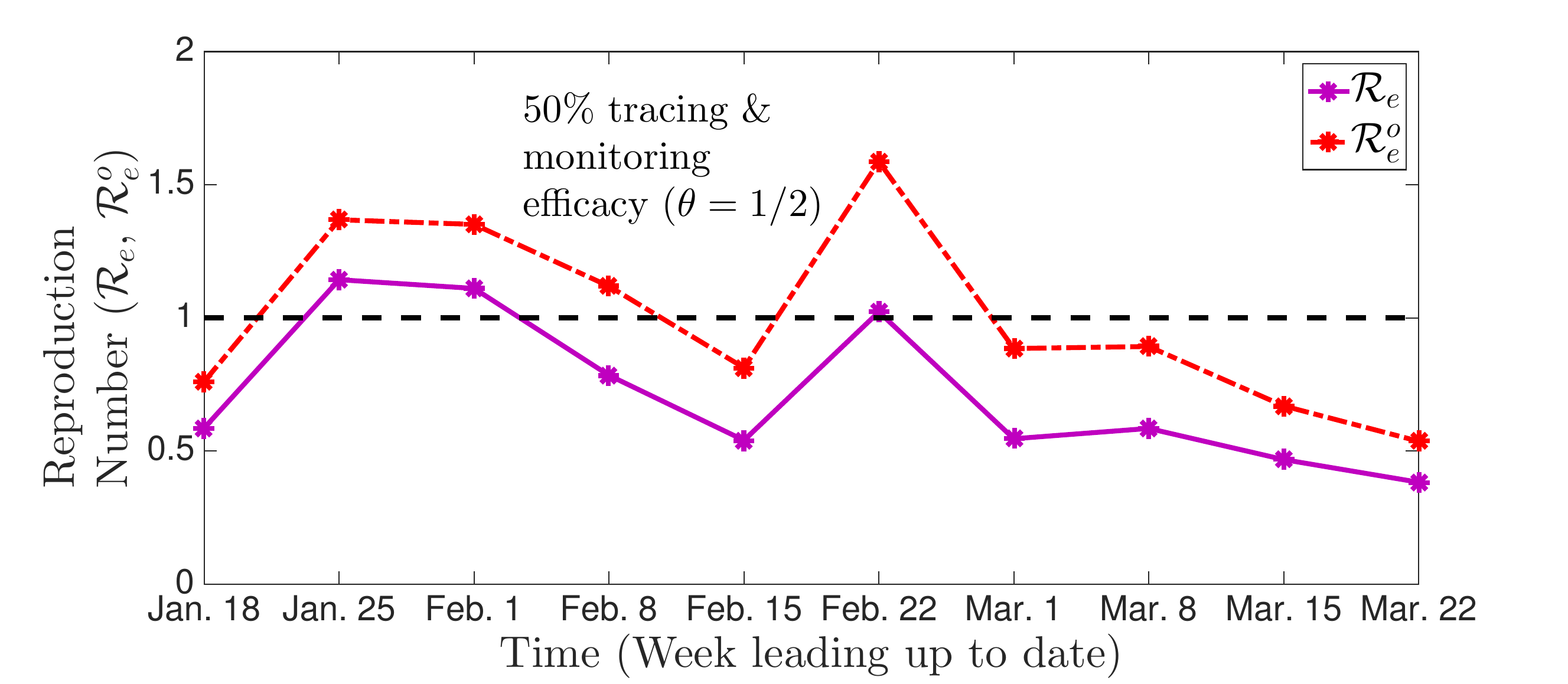}} \\
   \subfigure[][]{\label{G0} \includegraphics[width=.5\textwidth,height=.35\textwidth]{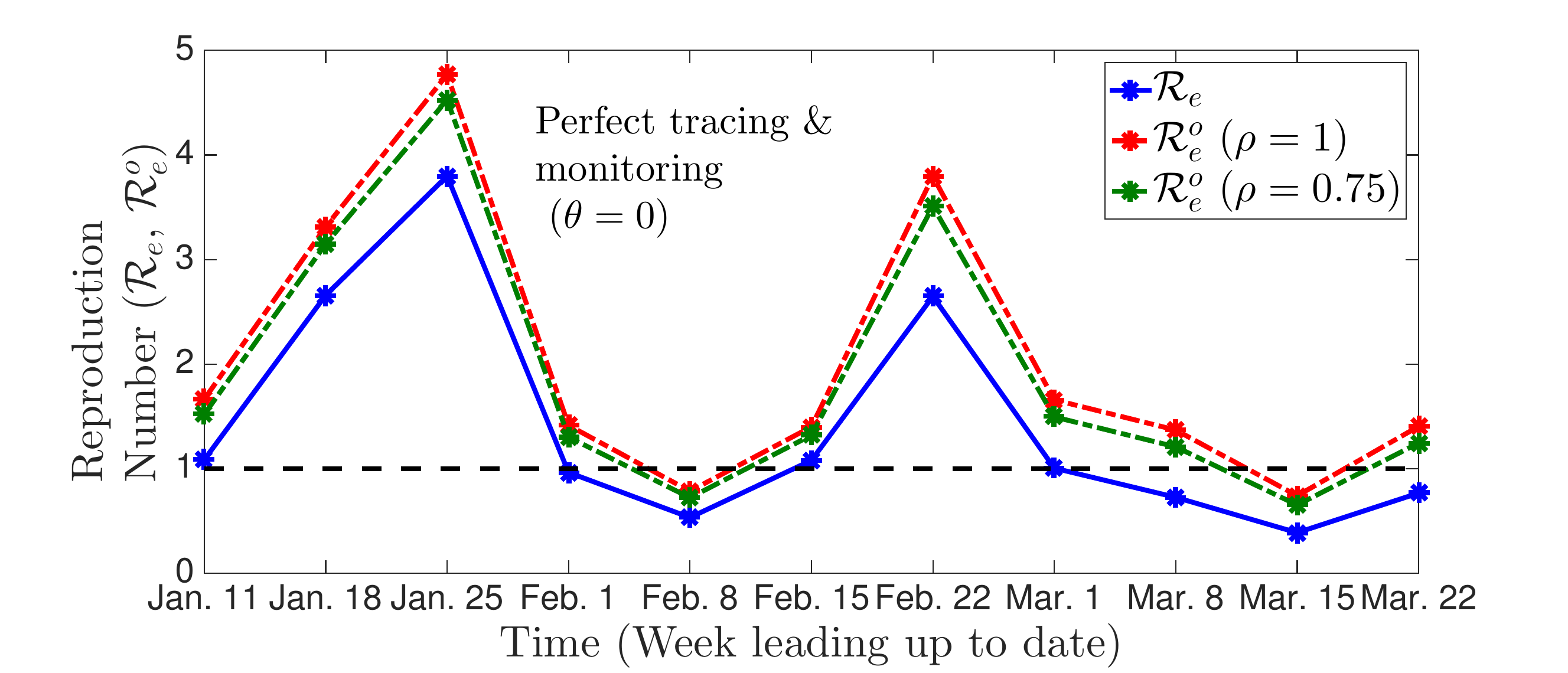}}
     \subfigure[][]{\label{S0} \includegraphics[width=.5\textwidth,height=.35\textwidth]{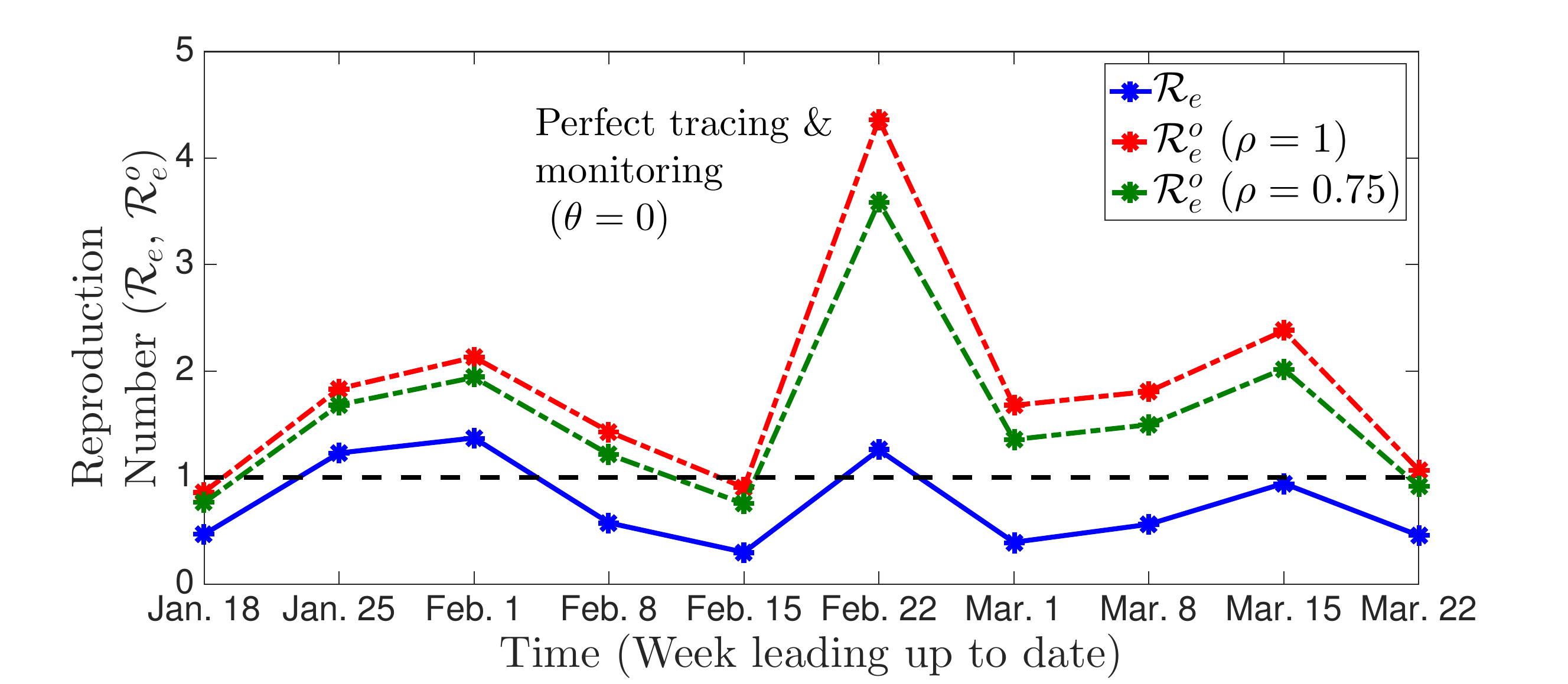}} \\
      \subfigure[][]{\label{G00} \includegraphics[width=.5\textwidth,height=.35\textwidth]{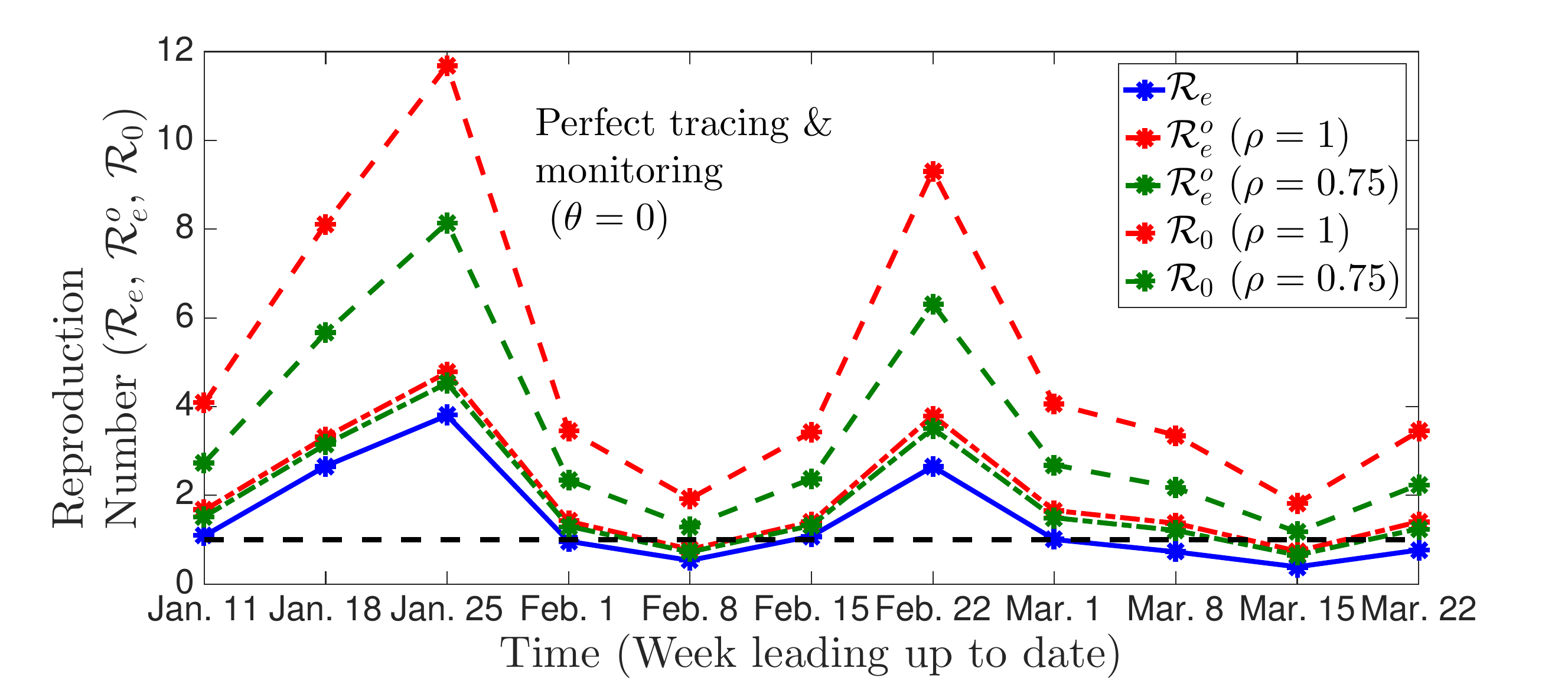}}
     \subfigure[][]{\label{S00} \includegraphics[width=.5\textwidth,height=.35\textwidth]{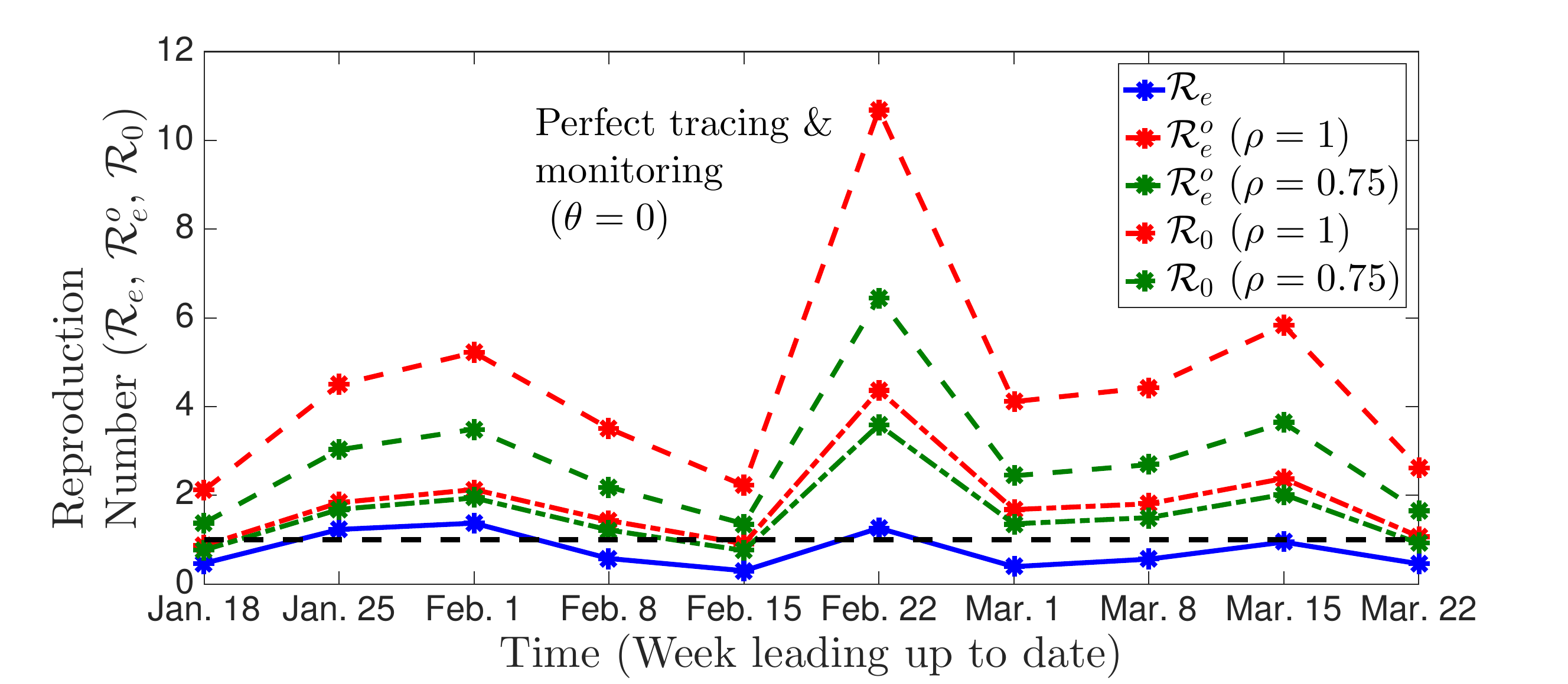}} 
   \caption{\emph{Impact of contact tracing}: Weekly estimated values of $\mathcal R_e^o$ (red) (reproduction number without tracing, calculated using (\ref{RelChange}) and (\ref{RelChange2}) ) are plotted above the estimated values of $\mathcal R_e$ for both Guinea (a,c,e) and Sierra Leone (b,d,f) assuming traced contact reproduction number reduction $\theta=0.5$ (a,b) (purple) and $\rho=1$; and perfect tracing, i.e. $\theta=0$ (c-f) (blue), for two different levels of reporting of untraced cases, $75 \% $ (green) and $100 \%$ (red).  In (e) and (f), the basic reproduction number, $\mathcal R_0$, without any case hospitalization is also displayed assuming $T_h=4$ and $T_u=9.8$. }
   \label{CTimpact}
   \end{figure} 
   
  Our methodology can quantify the impact of contact tracing and case hospitalization independent of tracing (passive case finding) in reducing the reproduction number utilizing formulas (\ref{RelChange}) and (\ref{RelChange2}).  In particular, under the assumption of perfect tracing ($\theta=0$), for given level of reporting of untraced cases (fraction $\rho$), the effective reproduction number \emph{without} contact tracing, $\mathcal R_e^o$, can be calculated with the weekly estimates of $\mathcal R_e$, the data for fraction of reported cases that are traced contacts ($q$), and formula (\ref{RelChange}).  In Figures \ref{G0} and \ref{S0}, $\mathcal R_e^o$ is plotted above $\mathcal R_e$ for two different levels of reporting untraced cases, $75 \% $ and $100 \%$, for both Guinea and Sierra Leone.  We also consider impact of contact tracing in the case of imperfect tracing (and 100 \% reporting), namely for tracing specific reproduction number reduction of $50 \%$ ($\theta=0.5$) in Figures \ref{G05} and \ref{S05}.  From the estimates of $\mathcal R_e$ and data for $q$, we also determine the reproduction number, $\mathcal R_0$, without contact tracing or any case hospitalization given values of $\rho$ and $T_h/T_u$ in the case of perfect tracing.  Figures \ref{G00} and \ref{S00}, compare weekly values of $\mathcal R_e$ with $\mathcal R_e^o$ and $\mathcal R_0$ for $\rho=0.75$ and $\rho=1$ (assuming $T_h=4$ and $T_u=9.8$).  Thus we can utilize formulas (\ref{q1}), (\ref{ReImperfect}), (\ref{RelChange}) and (\ref{RelChange2}) with contact tracing observables in order to determine the reduction in reproduction number due to contact tracing and passive case finding, which allows for evaluating the effectiveness of these different interventions. 
  
   These figures illustrate that contact tracing in both countries had an impact on $\mathcal R_e$ from January to the middle of March, but the impact has been much more dramatic in Sierra Leone.  Indeed, even when assuming $50 \%$ efficacy of contact tracing/monitoring ($\theta=0.5$), our calculations show that contact tracing has pushed the average effective reproduction number over this time period below one in Sierra Leone.  Our results indicate that contact tracing efforts may need to be increased in Guinea in order to ensure epidemic control and to reduce the risk from unknown chains of transmission.  The target goal of WHO is ``for 100\% of new cases to arise among registered contacts, so that each and every chain of transmission can be tracked and terminated''.    While this is certainly an ideal target, it is important to at least reduce and stabilize $\mathcal R_e$ substantially below one in order to control the epidemic with limited resources and end the outbreak in a timely manner.

\subsection{Impact of contact tracing and control on disease dynamics in general model}\label{SecImpact}

In order to investigate the coupled effects of contact tracing, case reporting and disease characteristics on epidemic spread, we analyze the general model (\ref{M2.1}).  In particular, variable tracing responsiveness and monitoring protocols, along with public health campaigns independent of contact tracing, are evaluated according to their impact on the effective reproduction number $\mathcal R_e$ and cumulative case count.  In addition, computation of $\mathcal R_e$ in this general setting helps to quantify potential underestimation of $\mathcal R_e$ in formulas derived in Section \ref{RuleOfThumb} for the simplifying cases of ``perfect'' tracing speed ($p_e=1$), and either perfect monitoring efficacy ($\beta_m=0$) or perfect reporting ($\rho=1$).   

The effective reproduction number $\mathcal R_e$ can be reduced by both contact tracing and public campaigns aimed at increasing case hospitalization and decreasing time until hospitalization among untraced individuals (passive case finding).  Considering both types of control strategies together, we find that passive case finding has a synergistic effect on the efficacy of contact tracing and reduction of $\mathcal R_e$.  In particular, Figure \ref{contour} shows that reduction in time until hospitalization (independent of tracing), $T_h$, has a nonlinear impact on the proportion of reported case contacts that are traced, $\phi$, required to bring $\mathcal R_e$ below one.  Indeed, reducing $T_h$ has the effect of both decreasing the reproduction number and number of contacts of untraced individuals and increasing the efficacy of tracing through improvements in tracing speed ($p_e$ is increased, $T_i$ is decreased).  In the \ref{AddFigs}, Figure \ref{RhoPhi} illustrates a similar synergistic effect observed for increasing case hospitalization/reporting, $\rho$.

   \begin{figure}[t!]       
\subfigure[][]{ \label{contour} \includegraphics[width=.5\textwidth,height=.35\textwidth]{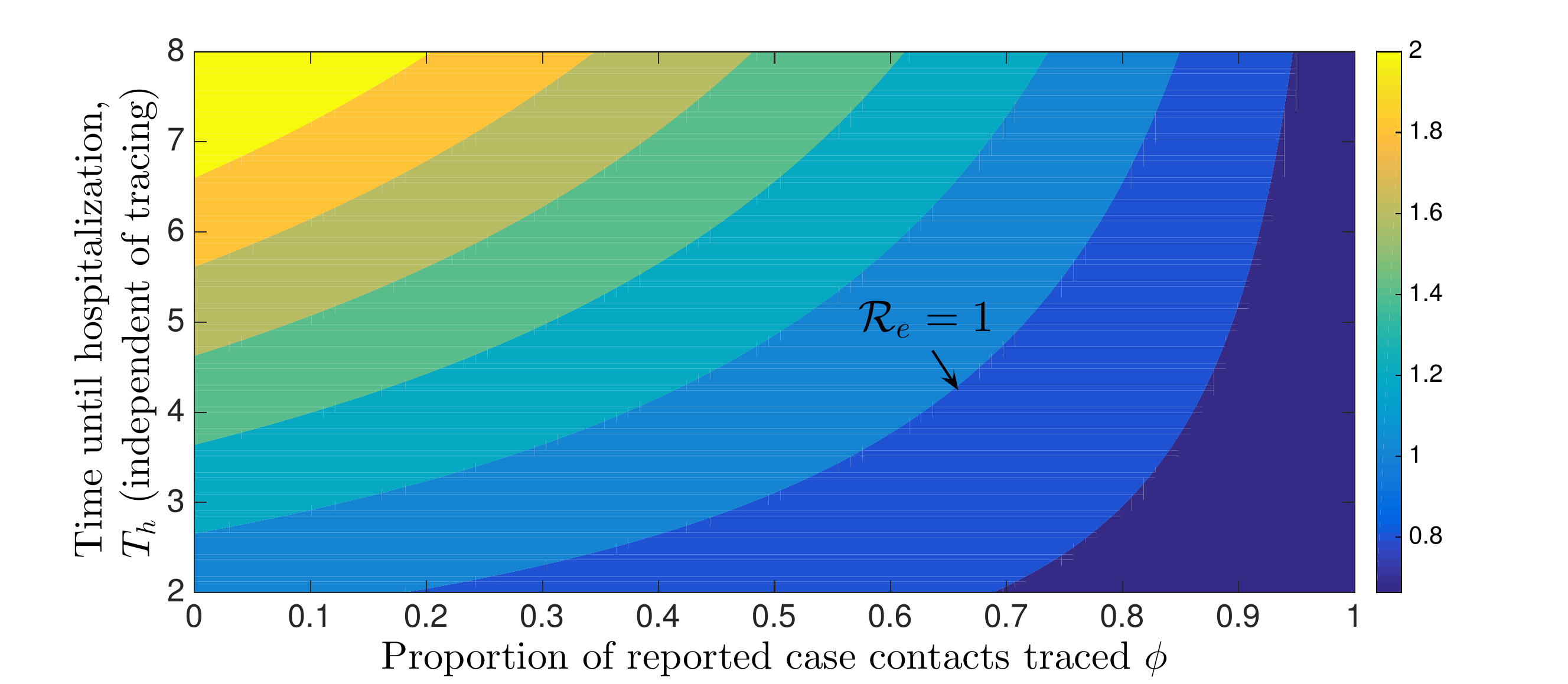}}
\subfigure[][]{ \label{beta_m} \includegraphics[width=.5\textwidth,height=.35\textwidth]{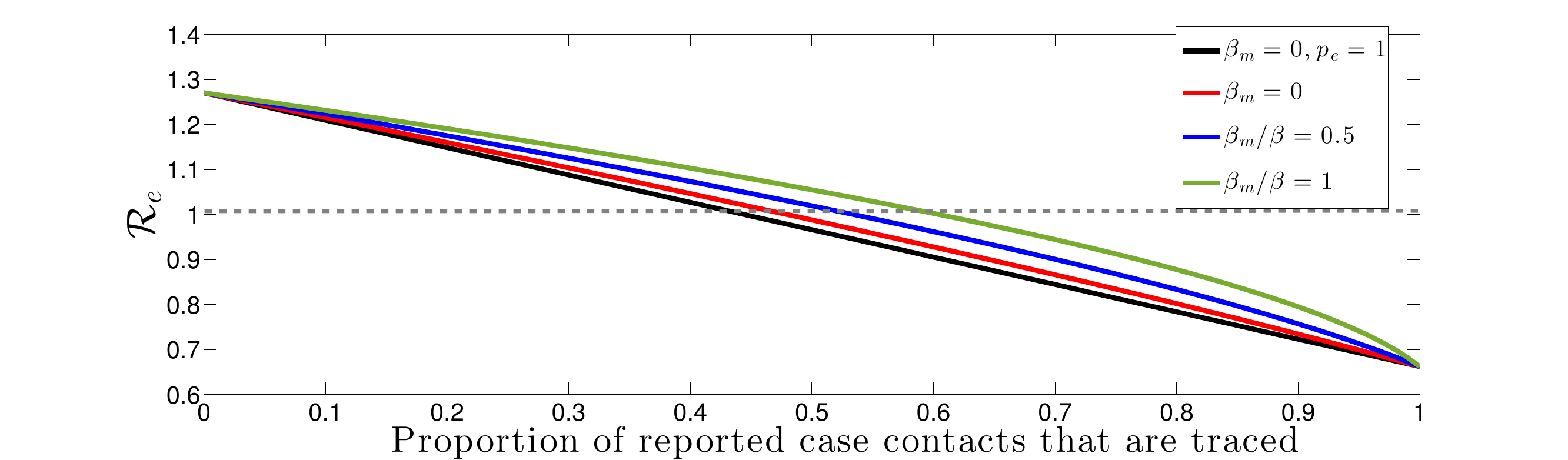}}
 \caption{ (a) Contour map of $\mathcal R_e$ when varying time until hospitalization/reporting ($T_h$) and proportion of reported case contacts that are traced ($\phi$) in the special case of perfect monitoring ($\beta_m=0$).  Note that for the other parameters, we set:  $\beta=0.4/N$ where $N= 6 \times 10^{6}$ (population of Sierra Leone), infectious period for unreported infected $T_u=9.8$, proportion of untraced infected hospitalized $\rho=0.75$, incubation period $\tau=10$ and tracking delay $\delta=2$.  Note $p_e$, $p_e^e$, $p_e^i$, $T_i$ are given by formulas derived earlier. (b) The effect of varying $\beta_m$ (monitoring efficacy) and $\phi$.  Note that $T_m=2$ and all other parameters are as in (a).}
 \end{figure}    
 
   \begin{figure}[t!]       
\subfigure[][]{ \includegraphics[width=.5\textwidth,height=.35\textwidth]{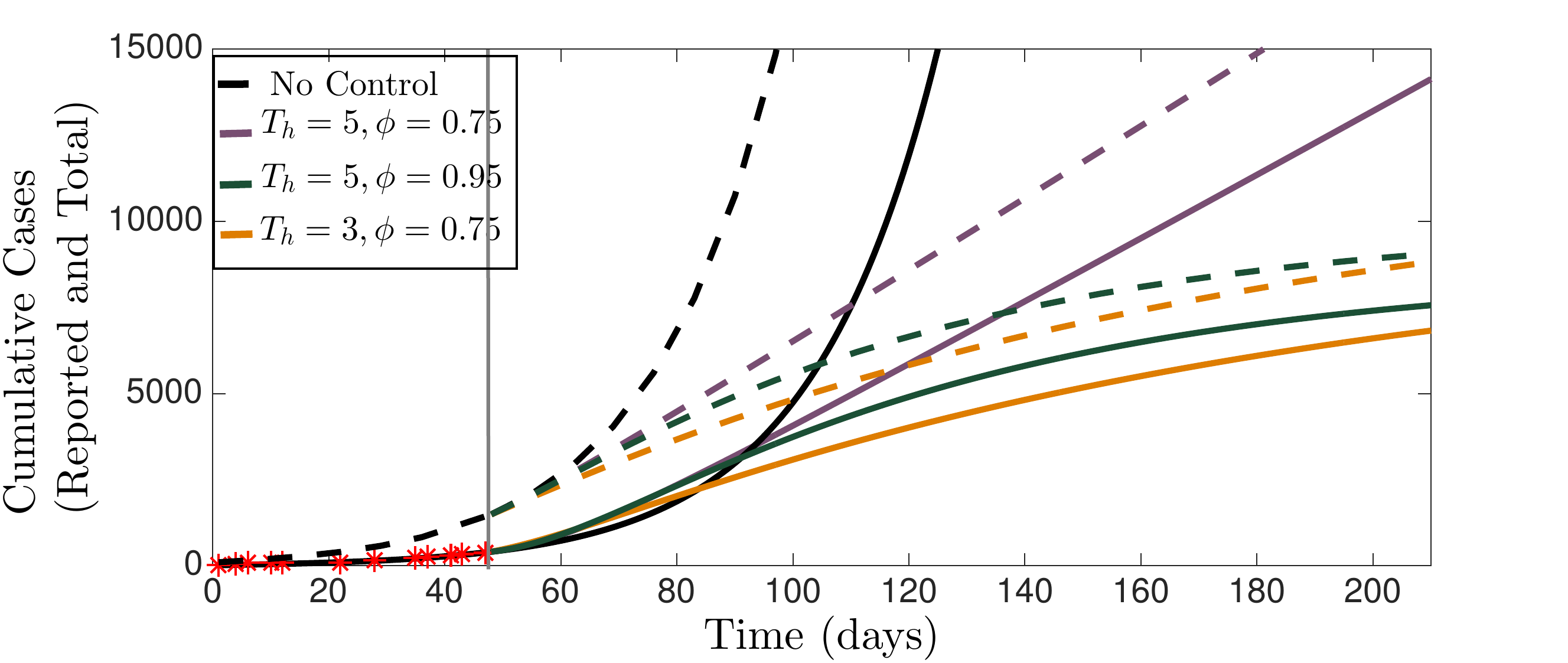}}
\subfigure[][]{ \includegraphics[width=.5\textwidth,height=.4\textwidth]{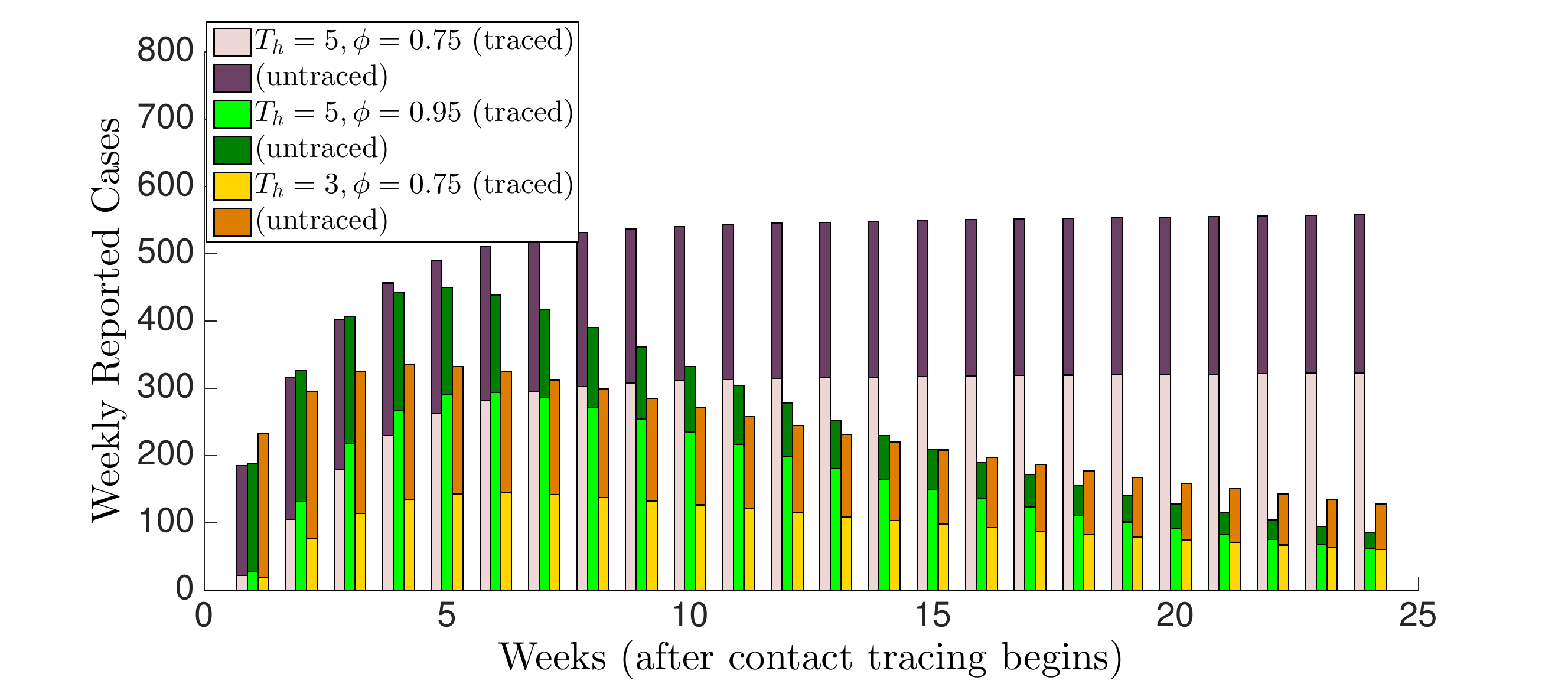}} \\
\subfigure[][]{\label{FracTrace} \includegraphics[width=.5\textwidth,height=.35\textwidth]{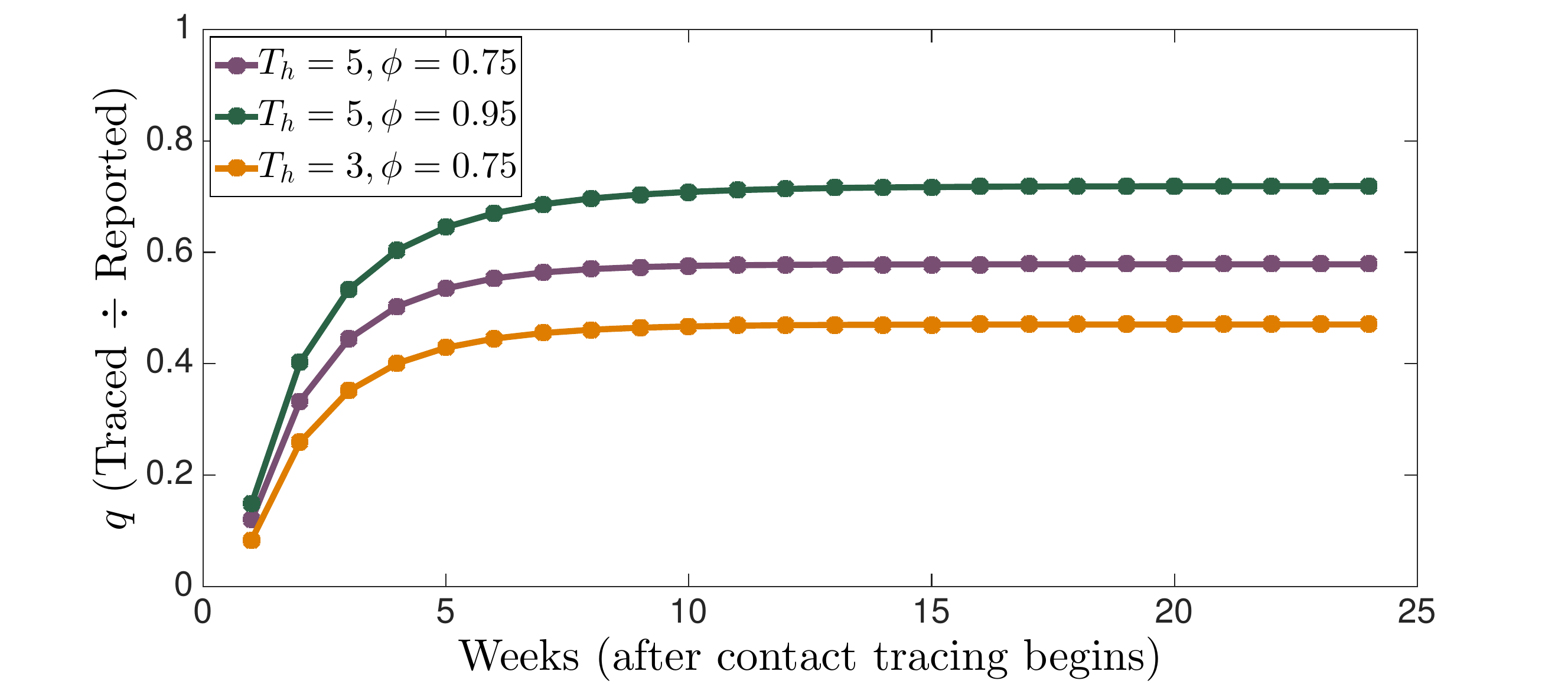}}
\subfigure[][]{\label{VaryStart} \includegraphics[width=.5\textwidth,height=.35\textwidth]{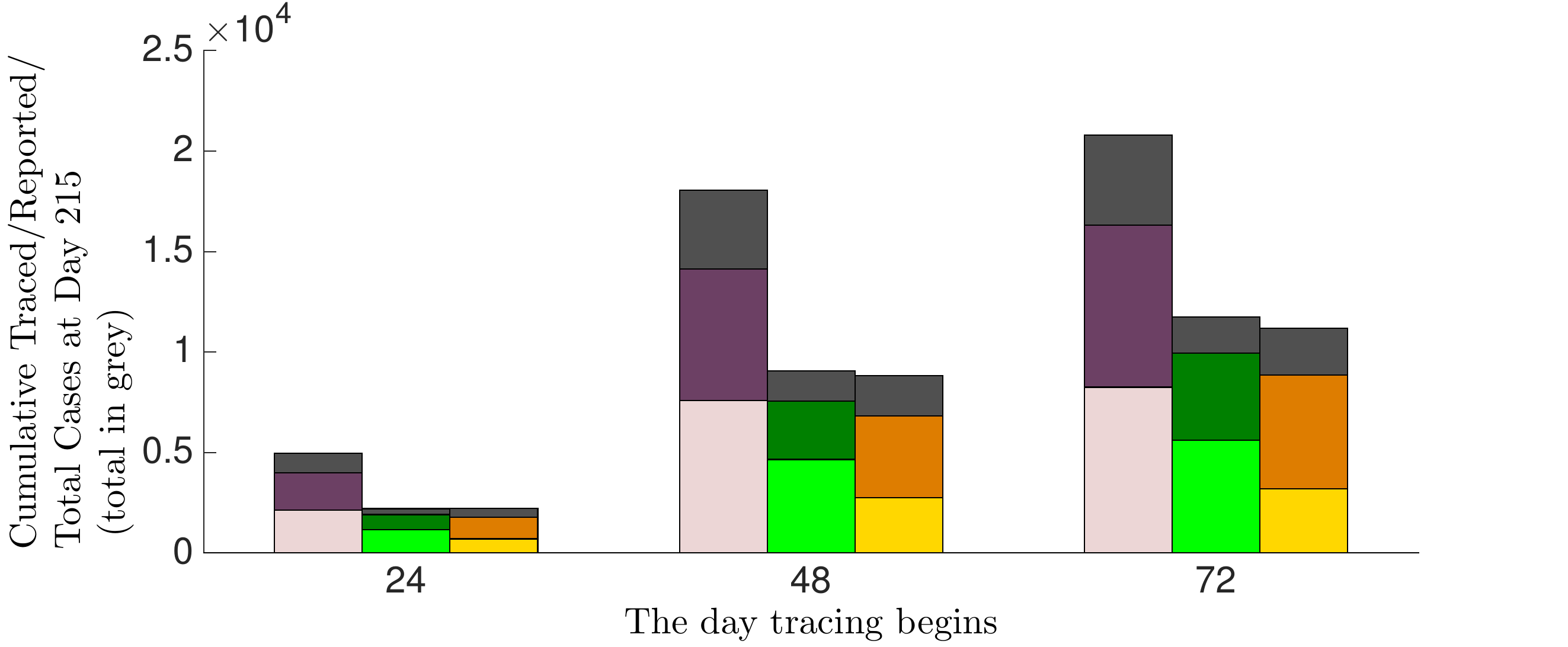}}
 \caption{ (a) Simulations of cumulative reported and total cases before and after contact tracing/enhanced control measures.  The before represents the initial exponential increase of the epidemic in Sierra Leone with data fit by model using parameters $\beta=0.4/N$ (as above), $\rho=0.5$, yielding $\mathcal R_e=2$.  After the vertical dashed line (48 days from $t=0$) shows the trajectory of reported (solid) and total (dashed) cumulative cases for four different scenarios: no control ($\mathcal R_e=2$); $T_h=5,\rho=0.75$ ($\mathcal R_e=1.0068$);$T_h=5,\rho=0.95$ ($\mathcal R_e=0.7531$);$T_h=3,\rho=0.75$ ($\mathcal R_e=0.8679$).  Also $\beta_m=0.5 \beta$, $\rho=0.75$, $\delta=2$, $T_m=2$ and $\tau=10$.    Note that actual cumulative reported cases for Sierra Leone is $\approx \ 9,500$ at end time of simulations ( Dec. 31). (b) Cumulative weekly reported and traced cases for the scenarios in (a).  (c)  Weekly fraction of reported cases which are traced contacts  (d) Cumulative total, reported and traced at day 215 for the three scenarios in (a) when tracing (and enhanced control) begins at day 24, 48 and 72. }
 \label{CumCases}
 \end{figure}

Consideration of the details of the contact tracing process leads to questions about how distinct tracing and monitoring parameters affect the disease dynamics.  We explore how the extent of contact tracing (through proportion of reported case contacts that are traced, $\phi$), strictness of monitoring protocol (transmission rate of traced contacts, $\beta_m$) and delay between index case reporting to tracing ($\delta$), affect $\mathcal R_e$ (Figure \ref{beta_m}, \ref{DeltaTau}) in \ref{AddFigs}).  We generally find that contact tracing more individuals has a larger impact then the strictness of monitoring and tracing speed. Indeed, increasing the fraction of contacts traced, $\phi$, is a more ``upstream'' intervention and the relatively long mean incubation period $\tau$ ($\approx 10$ days) of Ebola allows for more time to track contacts.  This long incubation period, combined with the fact that infectiousness and symptoms onset are coincidental and transmission occurs only through direct contact, are disease traits which make Ebola particularly amenable to contact tracing.  Our observations are corroborated in a global sensitivity analysis of $\mathcal R_e$ with respect to the model parameters (Figure \ref{ReSens} in \ref{AddFigs}).  Overall, this sensitivity analysis suggests that the proportion of (untraced) cases hospitalized/reported, $\rho$, has the largest impact on $\mathcal R_e$, followed by time until hospitalization (independent of tracing) $T_h$, transmission rate $\beta$ and proportion of reported case contacts traced $\phi$.

In addition to the effect on $\mathcal R_e$, we investigate the impact of parameters and various control scenarios on the cumulative case count in the model (\ref{M2.1}).  Assuming an initial reproduction number of $\mathcal R_e=2$ (with $T_h=5$, $T_u=9.8$, $\rho=0.5$) for the 48 days between May 27 and July 14 (a reasonable fit for the initial exponential epidemic growth in Sierra Leone \cite{Chowell, Team}), we simulate forward the cumulative cases under different control scenarios in Figure \ref{CumCases}.  The simulations show that minimizing effective reproduction number $\mathcal R_e$ does not always align with minimal cumulative cases over a time period (or over the entire epidemic; example shown in SI).  In particular, reducing time until hospitalization to $T_h=3$ with contact tracing proportion $\phi=0.75$ has a larger impact on cumulative cases over the time period $t=48$ to $t=210$ than the case when $T_h=5$, $\phi=0.95$, even though the respective $\mathcal R_e$ values have opposite ordering.  Also, the number of traced reported cases for the first scenario ($T_h=3,\phi=0.75$) is substantially less than the other scenario ($T_h=5,\phi=0.95$), which may be important given limited contact tracing resources.  While our model does not explicitly account for tracing susceptible individuals, the number of total traced individuals can be thought of as proportional to the number of traced infected with proportionality constant $1/p$ where $p$ is the probability of transmission per contact.  Simulations show that timing of the control measures are crucial in reducing both cumulative cases and the number of contacts which need to be traced (Figure \ref{VaryStart}).  Global sensitivity analysis of total cumulative cases at $t=148$ is similar to the sensitivity analysis of $\mathcal R_e$, except that $\phi$ has slightly less impact, and the tracing delay $\delta$ and incubation period $\tau$ (both of which affect the probability of reaching an infected contact during incubation period, $p_e$) gain more importance.  Summarizing the results, we find that passive case finding along with contact tracing, and early initiation of the interventions are key components in controlling an Ebola epidemic under the constraint of limited resources.  

To conclude this section, we remark upon how the simple formulas derived for $\mathcal R_e$ under the assumption of ``perfect'' tracing or perfect reporting differ from the more general case.  Clearly the case of ``perfect'' tracing ($\beta_m=0$, $p_e=1$) results in lower values of $\mathcal R_e$ than in the imperfect setting of $\beta_m>0$, $p_e<1$.  However, as shown in Figure \ref{beta_m}, the simplification of perfect tracing does not lead to substantially lower values of $\mathcal R_e$ than the imperfect case, especially in the general case of ``perfect monitoring'' ($\beta_m=0$, $p_e\leq 1$).  A more interesting question to explore is the following: how well does the formula in terms of tracing observables, $\mathcal R_e=k(1-q)/q+k_m$ (where $k$ and $k_m$ is $\#$ of infected traced per reported untraced and traced case, respectively,  and $q$ is fraction of reported cases that are traced), hold in the general setting?  In Figure \ref{FracTrace}, it is shown that the weekly fraction of reported cases that are traced, $q(t)$, converges to a constant value for each scenario.  Thus, a constant value $q$ seems to be possible to define more generally, but this quantity, along with $k$ and $k_m$, may be difficult to formulate in a simple expression of model parameters.   It is conceivable that such a formula can hold in general by a similar heuristic argument for the formula holding with underreporting given in Section \ref{RuleOfThumb} for the case of perfect tracing.  Indeed, while imperfect reporting and tracing lead to ``missing cases'' in individual transmission chains of an index reported case, more reported cases are generated throughout the network, potentially balancing out the effect of imperfect reporting and tracing.



  \FloatBarrier

\section{Discussion} \label{discussion}

During the 2014-2015 West Africa Ebola epidemic, it was clear that control measures were insufficient as the number of cases grew rapidly, but ultimately improvements in individual behavior, community engagement, early case isolation and safe burials have slowed down the disease trajectory. To have the best chance of averting another epidemic of Ebola or another emergent pathogen, it is important to know which particular interventions were most responsible for reversing the momentum. In this study, we focus on analyzing the effectiveness of contact tracing along with passive case finding for the West Africa Ebola outbreak and other potential epidemics in the future.

We have developed a mechanistic differential equation model of the contact tracing process in an infectious disease outbreak, particularly relevant to Ebola outbreaks.  It is widely recognized that contact tracing is an integral component of Ebola interventions \cite{WHOSit}.  Our deterministic model explicitly links contact tracing to reported case transmissions, and incorporates the coupled effects of disease traits, passive case finding (independent of tracing) and tracing/monitoring efforts.  The novel framework of the model is separating infectious cases that will be reported, further separating a fraction of their transmissions which will be traced and determining the action of tracers based on infectious stage of the traced contact.  This captures the dynamic structure of contact tracing.  


 The explicit linking to transmission in our model brings out an interesting connection between observable contact tracing data and the effective reproduction $\mathcal R_e$, yielding novel formulas for real-time estimates of $\mathcal R_e$.  We show that under the assumption of perfect tracking, along with either perfect monitoring of traced contacts or perfect reporting of all cases, $\mathcal R_e=k(1-q)/q+k_m$ where $k$ and $k_m$ are the average number of secondary infected individuals traced per primary reported untraced and trace case, respectively, and $(1-q)/q$ is the odds that a reported cases is not a traced contact.  This formula (\ref{q1}) can provide simple current estimates of $\mathcal R_e$ in the population, as illustrated utilizing a limited dataset of weekly reported and traced Ebola cases in Guinea and Sierra Leone from mid January to April 2015.  Our results indicate that contact tracing has likely been important in calming the epidemic for Sierra Leone over this time period, but had less impact on the disease dynamics in Guinea.  Increasing contact tracing effort and efficacy has the potential to eliminate the disease at this later stage in the outbreak.  The relative contribution of contact tracing in influencing the value of $\mathcal R_e$ is increasing with $q$, the fraction of reported cases that are traced contacts, and  $\mathcal R_e\rightarrow k_m$ as $q\rightarrow 1$.  Therefore, the goal that WHO has set forth to bring the fraction of traced reported $q$ to one in order to ``eliminate all unknown transmission chains'' is indicative of utilizing contact tracing to end the epidemic.  

There is less applicable data on contact tracing in 2014 when the Ebola outbreak was at a more critical stage in West Africa, and there are more questions about the efficacy of tracing during times of intense transmission as contact tracers may have been overwhelmed with high volumes of contacts to trace \cite{arwady2015evolution}.  While further investigation into the efficacy of contact tracing in West Africa throughout the epidemic utilizing all data sources is needed, this is beyond the scope of our current work.  However, we can utilize our sensitivity analysis and model simulations to gain insight on the impact and optimal deployment of contact tracing during different stages in the outbreak.  First, our study suggests that a combination of interventions can work synergistically, which may help to alleviate logistical problems with contact tracing in high transmission areas.  In particular, improvements in passive case finding (faster and more frequent isolation independent of contact tracing) together with contact tracing are most efficient in reducing incidence.  In this case, even with limited resources, contact tracing can play an important role in controlling hot zones.

Contact tracing is a multi-faceted strategy with multiple ways of impacting disease dynamics and distinct operational variables to consider.  For instance, tracking the contacts and notifying/educating them about their status can reduce their secondary transmission if they develop symptoms, even if they are not isolated right away.  Our model includes the details of the contact tracing process along with quantifying their relative contribution to the overall efficacy of the intervention.  We find that the fraction of (high-risk) contacts traced per reported case is more important than delays in tracking (since the incubation period of Ebola is relatively long) and potentially the strictness of monitoring protocol (assuming their transmission and time to isolation is reduced enough).   Optimizing the allocation of resources to tracing more individuals can possibly be achieved, for example, by checking up on traced contacts less often, but establishing trust and educating on the first visit to reduce the risk of that contact transmitting the disease.  Another finding from simulations of our model is the dramatic difference that earlier intervention can make when applied before the number of cases is high.  Contact tracing effectively at the beginning of an epidemic substantially reduces cumulative cases and overall need of resources.  Certainly, it also vastly increases the probability of ending an outbreak rapidly, as in the case of Nigeria and Senegal in 2014.   Finally,  contact tracing data can be utilized to evaluate the overall epidemic status and efficacy of interventions in real time in order to adjust efforts to the current situation.

There are several limitations to our model, along with promising directions for future work.  As mentioned in the previous paragraphs, limited resources in the face of an inflating epidemic can severely restrict the impact of contact tracing.  Although our results shed light on how to alleviate this problem through passive case finding or adjusting monitoring protocols, it may be informative to explicitly incorporate tracing logistics in the model.   Another future line of research will be to investigate whether formulas relating reproduction number to contact tracing observables, such as $\mathcal R_e=k(1-q)/q+k_m$, can be derived for contact networks, along with more generally extending our model to networks.   In addition, further work will involve evaluating the accuracy of these formulas through simulated data from stochastic versions of the model.  Also, other disease distributions, e.g. gamma-distributed incubation and infectious period, can be utilized to better match real data and generalize our approach.  Our work establishes a theoretical framework to build upon for the mechanistic modeling of contact tracing and utilizing contact tracing observables to obtain simple real-time estimates of effective reproduction number, along with evaluating the impact of control strategies on an epidemic. 

 \section*{Acknowledgement}
 We would like to thank Joshua Weitz for helpful comments and suggestions which improved this work.  We would also like to thank Leonid Bunimovich for a helpful discussion about auto-correlation in the data.  Additionally, we thank the anonymous reviewers and editor for comments which have led to an improved manuscript.

\appendix
\setcounter{figure}{0} 
\section{Calculation of effective reproduction number $\mathcal R_e$}\label{A1}

The effective reproduction number $\mathcal R_e$ can be defined utilizing the next-generation approach \cite{VanDen}.    First define the feasible region for the system (\ref{M2.1}) as 
$$\Gamma=\left\{\mathbf{x}=(S,E,E_c^e,E_c^i,I_h,I_u,I_c^e,I_c^i)^T\in\mathbb{R}^8_+ \ | \ N:=\sum\limits_{i=1}^8\mathbf{x}_i\leq N_0 \right\}, $$
where $\mathbb{R}^8_+$ denotes the non-negative orthant of $\mathbb{R}^8$ and $N_0$ is the initial total population size (of all compartments).  
We note that the system (\ref{M2.1}) is quasi-positive, and thus its solutions remain non-negative when their initial values are nonnegative. Summing the right-hand sides of (\ref{M2.1}), we find that $N'(t)\leq -\gamma\left(I_h(t)+I_u(t)+I_c^e(t)+I_c^i(t) \right) \leq 0$ for a positive constant $\gamma$.  Thus the solutions the system (\ref{M2.1}) remain in 
$\Gamma$ when their initial values are in $\Gamma$.
Notice that for initial susceptible population, $S_0$, $\mathcal E_0:=(S_0,0,0,0,0,0,0,0)^T$ is disease-free equilibrium of system (\ref{M2.1}).

We define a next-generation matrix by considering the linearized system at the disease-free equilibrium, $\mathcal E_0$.  Write the linearized ``infection'' sub-system as $\mathbf{y}'=(F-V)\mathbf{y}$, where $F$ contains entries corresponding to new infections, and $-V$ contains all other transition terms in the Jacobian matrix evaluated at $\mathcal E_0$.   Thus, we consider the following matrices:
$$F=S_0\begin{pmatrix} 0 & 0 & 0 & (1-\phi)\beta & \beta & (1-\phi)\beta_m& (1-\phi)\beta \\ 0 & 0 & 0 & p_e\phi \beta & 0 & p_e^e \phi \beta_m & p_e^i \phi \beta \\ 0 & 0 & 0 & (1-p_e) \phi \beta & 0 & (1-p_e^e) \phi \beta_m & (1-p_e^i) \phi \beta \\ 0 & 0 & 0 & 0 & 0 & 0 & 0 \\ 0 & 0 & 0 & 0 & 0 & 0 & 0 \\ 0 & 0 & 0 & 0 & 0 & 0 & 0 \\ 0 & 0 & 0 & 0 & 0 & 0 &0 \end{pmatrix},$$
$$V=\begin{pmatrix} \frac{1}{\tau} & 0 & 0 & 0 & 0 & 0 & 0 \\ 0 & \frac{1}{\tau} & 0 & 0 & 0 & 0 & 0 \\ 0 & 0 & \frac{1}{\tau} & 0 & 0 & 0 & 0\\ -\rho \frac{1}{\tau} & 0 & 0 & 1/T_h & 0 & 0 & 0 \\-(1-\rho) \frac{1}{\tau} & 0 & 0 & 0 & 1/T_u & 0 & 0 \\ 0 & -\frac{1}{\tau} & 0 & 0 & 0 & 1/T_m & 0 \\ 0 & 0 & -\frac{1}{\tau} & 0 & 0 & 0 & 1/T_i \end{pmatrix}.$$ 
The next-generation matrix describing expected number of new infections (by the different types of infectious cases) is then defined as $K:=FV^{-1}$ \cite{VanDen}.  The effective reproduction number, $\mathcal R_e$, is the spectral radius of $K$, $\varrho(K)$:
\begin{align*}
\mathcal R_e=\varrho(K)=\varrho(FV^{-1})
\end{align*} 

\subsection{Calculation of $T_i$ (\ref{Ti})}
For individuals in the $I_c^i$ compartment, the \emph{average infectious period (prior to being tracked)}, $T_i$, can be approximated by considering the expected value of the time after symptom onset when a traced contact is tracked (given that they have progressed past their incubation period upon tracking), $\hat{T}_i$.  The formula for $\hat{T}_i$ in the calculation of $T_i$ (\ref{Ti}) is derived as follows:
{\small{
 \begin{align*}
 \hat{T}_i& =  T_h+\delta -  \mathbb{E}\left( \ \text{length of incubation period } X \ | \ \text{infectious at tracking time} \right) \\
 &= T_h+\delta - \frac{1}{T_h} \int_0^{T_h}  \mathbb{E}\left( \ X \ | \ \text{infected at time } r\in [0,T_h] \ \text{and } X\in [r,T_h+\delta]  \ \right) \, dr \\
&= \frac{1}{T_h} \int_0^{T_h} \left[ T_h+\delta - \left(\int_{r}^{T_h+\delta} \frac{x e^{- (x-r)/\tau}}{\int_0^{T_h-r+\delta} e^{- s/\tau} \, ds} \, dx\right) \right] \, dr
\end{align*}
}}
Here the notation $\mathbb E\left( X | \mathcal A \right)$ denotes the expected value of random variable $X$ given information $\mathcal A$.  Note that the compartment $I_c^i$ includes both individuals that  are tracked when infectious and those that are tracked after infectiousness (removed by recovery, death or reporting independent of tracing).  We also remark that clearly the waiting times corresponding to infectious period independent of tracing (with mean approximated by $\bar{T}$) and contact tracking of infectious (with mean approximated by $\hat{T}_i$) are independent.  Thus it is reasonable to treat the corresponding processes as \emph{independent exponentially distributed}, and therefore to consider the sum of the rates in formula (\ref{Ti}).  Also, note that $\hat{T}_i$ can more accurately be calculated by also integrating over the index case infectious period probability distribution as with the calculation of $p_e$, however to simplify, we assume a fixed infectious period $T_h$ for the index reported case. 

\section{Formulae for $\mathcal R_e$ in special cases} \label{ReAppendix}

\noindent For the special cases of the model considered in Section \ref{RuleOfThumb}, consider the scenario in which traced individuals are always tracked during their incubation phase, i.e. $p_e=1$.  Then, we can reduce the general system (\ref{M2.1}) to the following system:
\begin{align*}
S' \ \ & = \ \ -\beta S(I_h+I_u)-\beta_m S I_c \\
E'\ \ &= \ \ \beta SI_u+ \left(1-\phi \right) \beta S I_h + \left(1-\phi \right) \beta_m S I_c - \frac{1}{\tau} E \\
I_h'  \ \ &= \ \ \frac{\rho}{\tau} E -\frac{1}{T_h}I_h \\
I_u' \ \ &= \ \ \frac{1-\rho}{\tau} E -\frac{1}{T_u} I_u \\
(E_c)'\ \ &= \ \ \phi \beta S I_h +\phi\beta_m S I_c - \frac{1}{\tau} E_c \\
(I_c)' \ \ &= \ \ \frac{1}{\tau} E_c -\frac{1}{T_m}I_c  \label{M2.1reduced} \tag{A2} \\
\end{align*}
Here $E_c(t)$ and $I_c(t)$ are the exposed individuals who are contact traced (necessarily in their incubation phase since $p_e=1$) and infectious contact traced (monitored) individuals, respectively.  

First, consider the scenario of ``perfect tracing'', i.e. the case of perfect monitoring ($\beta_m=0$).  In this case, since $\beta_m=0$, the traced individuals do not transmit the infection, thus $E_c$ and $I_c$ compartments are decoupled from the first four equations in system (\ref{M2.1reduced}).  Either by inserting $\beta_m=0$ and $p_e=1$ into the $F$ matrix in the next-generation decomposition for the general model or calculating $\mathcal R_e$ in the reduced model (\ref{M2.1reduced}) with $\beta_m=0$, we find:
\begin{align*}
\mathcal R_e= \mathcal R_0\left(\rho\frac{T_h}{T_u}(1-\phi) + 1-\rho \right) \label{Rereduced} \tag{A3},
\end{align*}
where $\mathcal R_0=\beta S_0 T_u$.

Consider $\frac{1}{T_m}\int_{t-a}^t{I_c(s)ds}$ and $\frac{1}{T_h}\int_{t-a}^t{I_h(s)ds}$, the cumulative number of reported contact traced cases and reported cases independent from contact tracing, respectively, from time $t-a$ to $t.$  Define the fraction of reported cases that are traced contact (from time $t-a$ to $t$):
\begin{align*}
\tilde{q}_a(t)=\frac{\frac{1}{T_m}\int_{t-a}^t{I_c(s)ds}}{\frac{1}{T_m}\int_{t-a}^t{I_c(s)ds}+\frac{1}{T_h}\int_{t-a}^t{I_h(s)ds}}.
\end{align*}
 In addition, define the following constant $q$:
 \begin{align*}
q=\frac{\phi}{\phi+(1-\phi)\rho+(1-\rho)\frac{T_u}{T_h}}.
\end{align*}
As mentioned in Section \ref{RuleOfThumb}, $q$ can be interpreted as the conditional probability that a transmission is traced given that this case is reported.   Indeed, the probability that a case (transmission) is reported can be expressed as the sum of the probabilities of the three mutually exclusive events that the transmission was caused by a traced reported case, untraced reported case, and unreported case.  These probabilities are the summands in the denominator (multiplied by $\rho$) of $q$.
\begin{proposition} For sufficiently small $a$ and sufficiently large time $t,$
\begin{align*}
\tilde{q}_a(t) \approx q.
\end{align*}
\end{proposition}
\begin{proof}
Note that
\begin{align*}
\tilde{q}_a(t) \ &=\ \frac{\frac{1}{\tau}\int_{t-a}^t{E_c(s)ds}-I_c(t)+I_c(t-a)}{\frac{1}{\tau}\int_{t-a}^t{E_c(s)ds}-I_c(t)+I_c(t-a)+\frac{\rho}{\tau}\int_{t-a}^t{E(s)ds}-I_h(t)+I_h(t-a)}\\
    \ &= \ \frac{\int_{t-a}^t{\phi\beta S(s)I_h(s) ds}+\psi_C(t)}{\int_{t-a}^t{\phi\beta S(s)I_h(s) ds}+\psi_C(t)+\rho\int_{t-a}^t{(\beta S(s)I_u(s)+(1-\phi)\beta S(s) I_h(s))ds}+\psi(t)}\\
		\ &= \ \frac{\phi \int_{t-a}^t{\beta S(s)I_h(s) ds}+\psi_C(t)}{(\phi+\rho(1-\phi))\int_{t-a}^t{\beta S(s)I_h(s)ds}+\rho\int_{t-a}^t{\beta S(s)I_u(s)ds+\psi(t)+\psi_C(t)}},\\
\end{align*}
where $\psi(t)= I_h(t-a)-I_h(t)+E(t-a)-E(t-a)$ and $\psi_C(t)= I_c(t-a)-I_c(t)+E_c(t-a)-E_c(t-a)$.\\
Moreover
\begin{align*}
\int_{t-a}^t{(\rho I'_u-(1-\rho)I'_h)ds}=\int_{t-a}^t{(\frac{1-\rho}{T_h}I_h-\frac{\rho}{T_u}I_u)ds}.
\end{align*}
Then we obtain the following equality:
\begin{align*}
T_u\rho(I_u(t)-I_u(t-a))+(1-\rho)(I_h(t-a)-I_h(t))=\int_{t-a}^t{(1-\rho)\frac{T_u}{T_h}I_h(s)ds}-\int_{t-a}^t{\rho I_u(s)ds}.
\end{align*}
Let define $\theta(t)=T_u\rho(I_u(t)-I_u(t-a))+(1-\rho)(I_h(t-a)-I_h(t)).$\\
For small enough $a$, we can assume $S(t)$ is constant over the interval $[t-a,t].$ Then
\begin{align*}
\tilde{q}(t)=\frac{\beta S \phi \int_{t-a}^t{I_h(s)ds}+\psi_C(t)}{\beta S (\phi+\rho(1-\phi) )\int_{t-a}^t{I_h(s)ds}+\beta S(1-\rho)\frac{T_u}{T_h}\int_{t-a}^t{I_h(s)ds}+\psi(t)+\psi_C(t)-\theta(t)}
\end{align*}
As in other outbreak models, it can be shown that the infection components converge to zero as $t\rightarrow \infty$.  Thus $\psi(t), \psi_C(t),\theta(t) \approx 0,$ when $t$ is sufficiently large. Then for sufficiently large $t,$ we obtain
\begin{align*}
\tilde{q}_a(t) \approx \frac{\phi}{\phi+(1-\phi)\rho+(1-\rho)\frac{T_u}{T_h}}=q.
\end{align*}
\end{proof}

Thus inserting $q$ and $k= \beta S_0 T_h \phi$ into (\ref{Rereduced}),
 we find the following formula for the effective reproduction number $\mathcal R_e$:
 \begin{align*}
\mathcal R_e&=\beta\left(\rho T_h(1-\phi) + (1-\rho)T_u \right)  \\
&= \frac{k}{\phi} \left( \rho(1-\phi) +1-\rho \right) \\
 &= k\left( \frac{1-q}{q} \right).
 \end{align*}
 
 Next, consider the case of imperfect tracing (monitoring), $\beta_m>0$, but perfect reporting, $\rho=1$.  As in the previous case, $\mathcal R_e$ reduces to a simple formula:
  \begin{align*}
\mathcal R_e=(1-\phi) \mathcal R +\phi \mathcal R_m,
\end{align*}
where  $\mathcal R:=\beta S_0 T_h$ and $\mathcal R_m=\beta_m S_0 T_m$.  In this case, it is not hard to see that $q=\phi$ and we can similarly prove that $\tilde{q}_a(t)\approx q$ for large enough $t$ and small enough $a$.  Then, with $k_m=\mathcal R_m \phi$, we have 
 \begin{align*}
\mathcal R_e
 &= k\left( \frac{1-q}{q} \right) +k_m.
 \end{align*}

\section{Base SEIR models} \label{A2}
Base model with single infectious compartment:
\begin{align*}
S'(t)&=-\beta S(t) I(t) \\
E'(t)&=\beta S(t) I(t) - \frac{1}{\tau} E(t) \\
I'(t)&=\frac{1}{\tau} E(t) -\frac{1}{T} I(t) \label{B2}\tag{B1} \\
H'(t)&=\frac{\rho}{T} I(t) 
 \end{align*}
 
 \begin{proposition}
 The two base models (\ref{B1}) and (\ref{B2}) are equivalent if $T_h=T_u=T$ and $I_h(0)=\rho I(0), I_u(0)=(1-\rho)I(0)$.  In addition, the reproduction number of (\ref{B2}) is $\mathcal R_0=\beta T$, thus the reproduction number of models (\ref{B1}) and (\ref{B2}) are equal in general when $T=\rho T_h + (1-\rho) T_u$.
 \end{proposition}
 \begin{proof}
 For the solution to (\ref{B1}), define $\hat{I}(t)=I_h(t)+I_u(t)$.  Then $\hat{I}'(t)=\frac{1}{\tau} E(t) -\frac{1}{T} \hat{I}(t) $ and $\hat{I}(0)=I_h(0)+I_u(0)=I(0)$.  Noting that the compartment $H$ is decoupled in both systems, we can conclude that $(S(t),E(t),\hat{I}(t))=\left(S(t),E(t),I(t)\right)$, where $I(t)$ is the infectious solution to (\ref{B2}).  Notice that $(\rho\hat{I})'=\frac{\rho}{\tau} E-\frac{1}{T}(\rho \hat{I})$ and $(\rho\hat{I})(0)=I_h(0)$.  Thus $I_h(t)=\rho \hat{I}(t)=\rho I(t)$.  Therefore, we also find that $H'(t)=\frac{1}{T}I_h(t)=\frac{\rho}{T}\hat{I}(t)=\frac{\rho}{T}I(t)$.  This shows models (\ref{B1}) and (\ref{B2}) are equivalent if $T_h=T_u=T$ and $I_h(0)=\rho I(0), I_u(0)=(1-\rho)I(0)$.  The second statement is immediate.
 \end{proof}
 Note that numerical simulations show that the final size of the epidemics match for solutions to models (\ref{B1}) and (\ref{B2}) when the reproduction numbers are equal, i.e. $T=\rho T_h + (1-\rho) T_u$ and $\beta$ is the same.  In addition, numerics show that if $T_h\approx T_u$, then the transient behavior of solutions of the two models ($I(t)$ and $I_h(t)+I_u(t)$ are almost identical, however as $|T_h-T_u|$ increases, the transient dynamics become more distinct even though the final sizes will match.

\section{Additional Sensitivity Analysis Figures} \label{AddFigs}
\FloatBarrier
     
  \begin{figure}[h]       
\subfigure[][]{ \label{RhoPhi} \includegraphics[width=.5\textwidth,height=.35\textwidth]{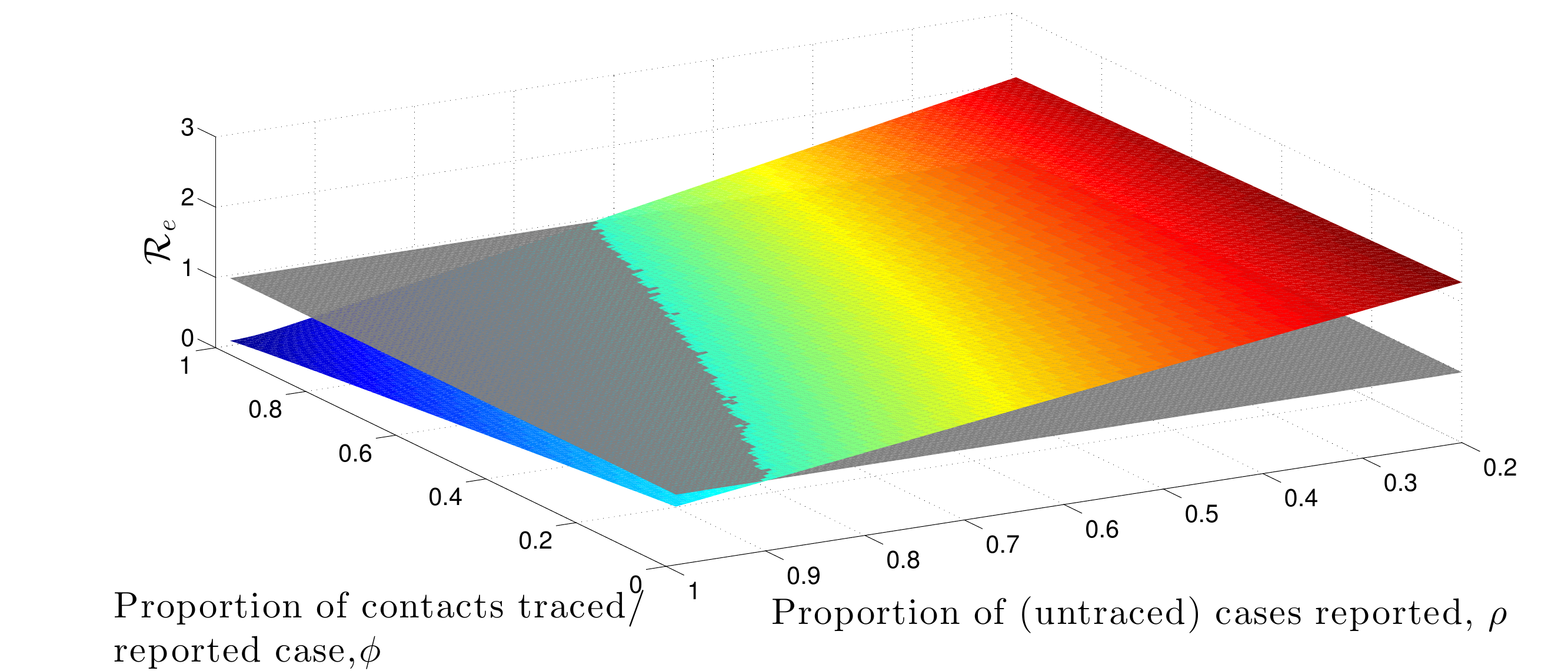}}
\subfigure[][]{ \label{DeltaTau} \includegraphics[width=.5\textwidth,height=.35\textwidth]{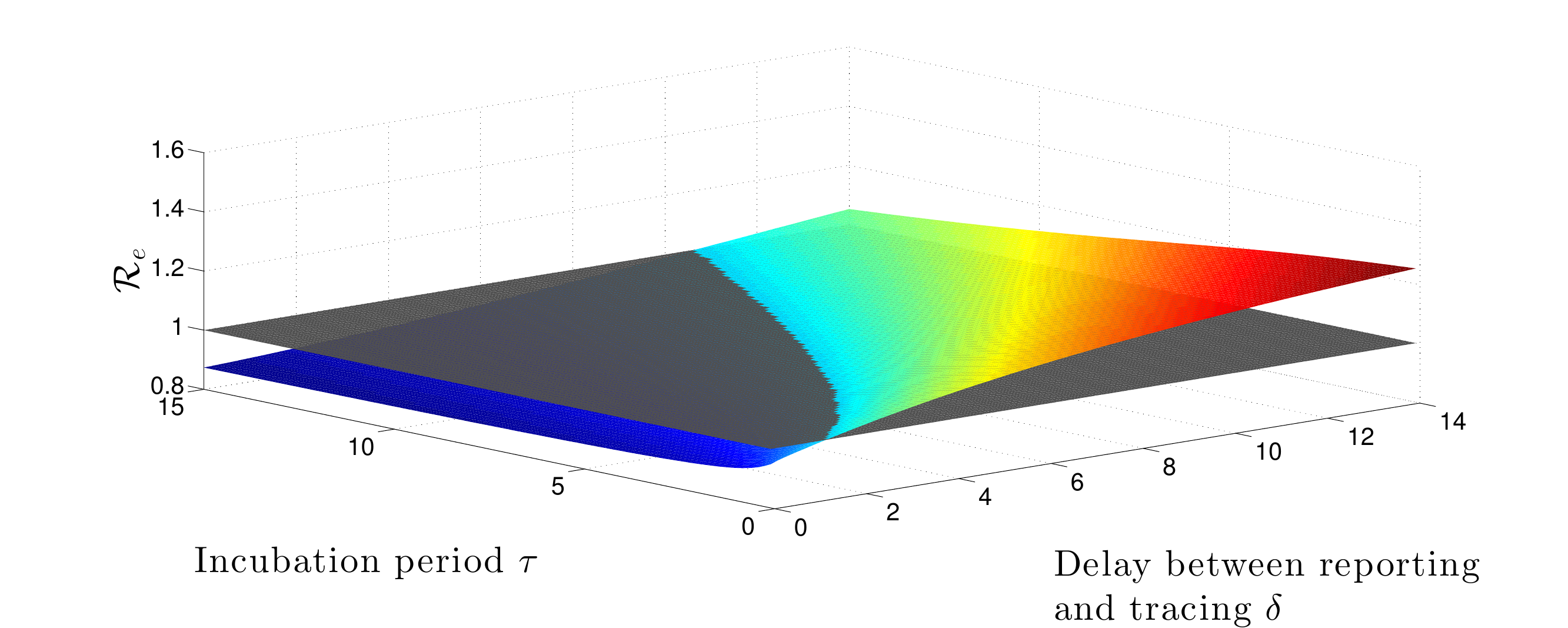}}
 \caption{ Surface of effective reproduction number, $\mathcal R_e$, for general model (\ref{M2.1}) when varying (a) $\rho$ and $\phi$, (b) $\delta$ and $\tau$.  The parameters are as in Figure \ref{contour} with $T_h=3$ and $\phi=\rho=0.75$ in (b).  }
 \label{AddSensFig}
 \end{figure}    
 
  \begin{figure}[h]       
 \subfigure[][]{ \includegraphics[width=.45\textwidth,height=.25\textwidth]{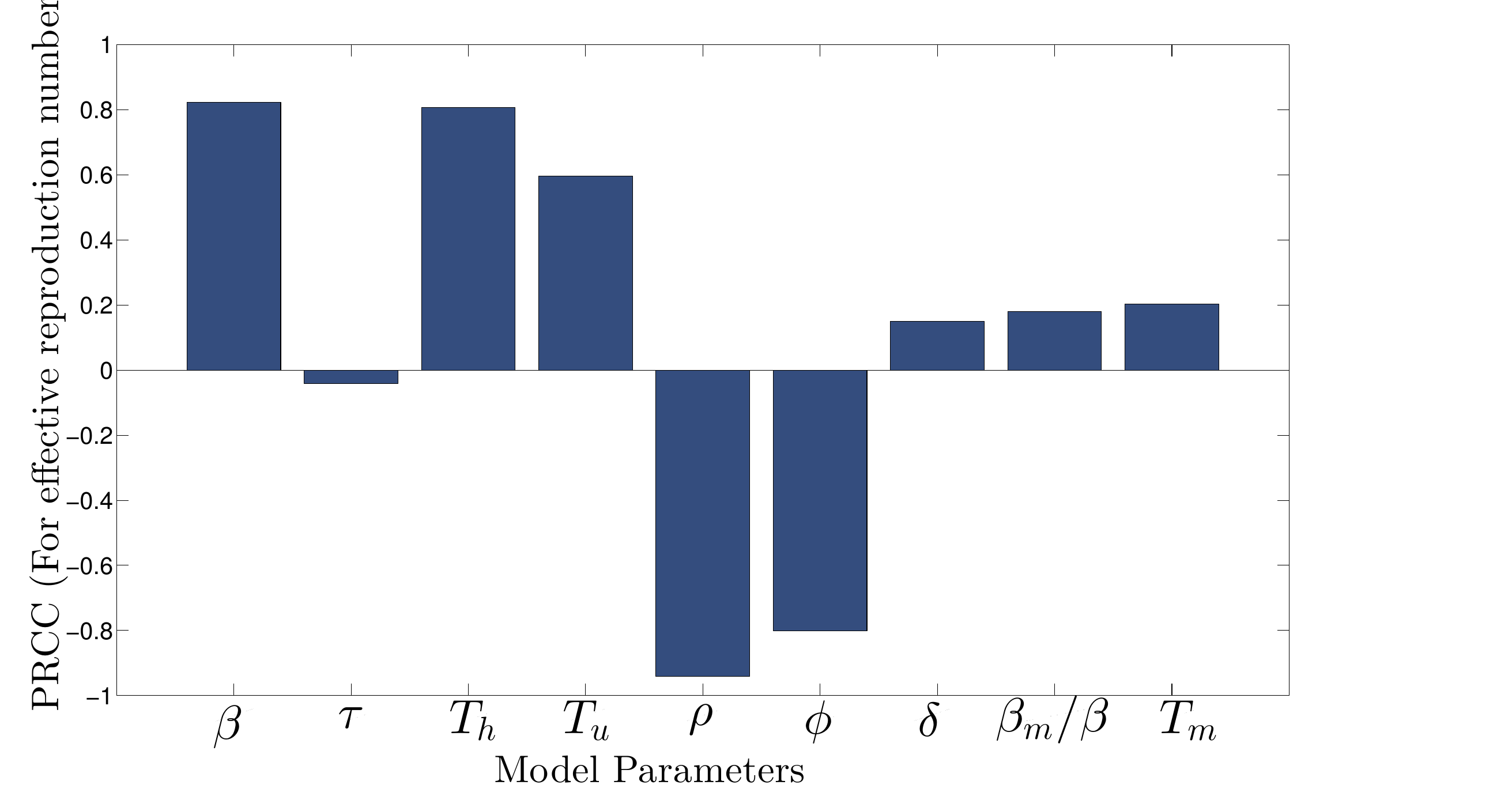}}
\subfigure[][]{ \includegraphics[width=.45\textwidth,height=.25\textwidth]{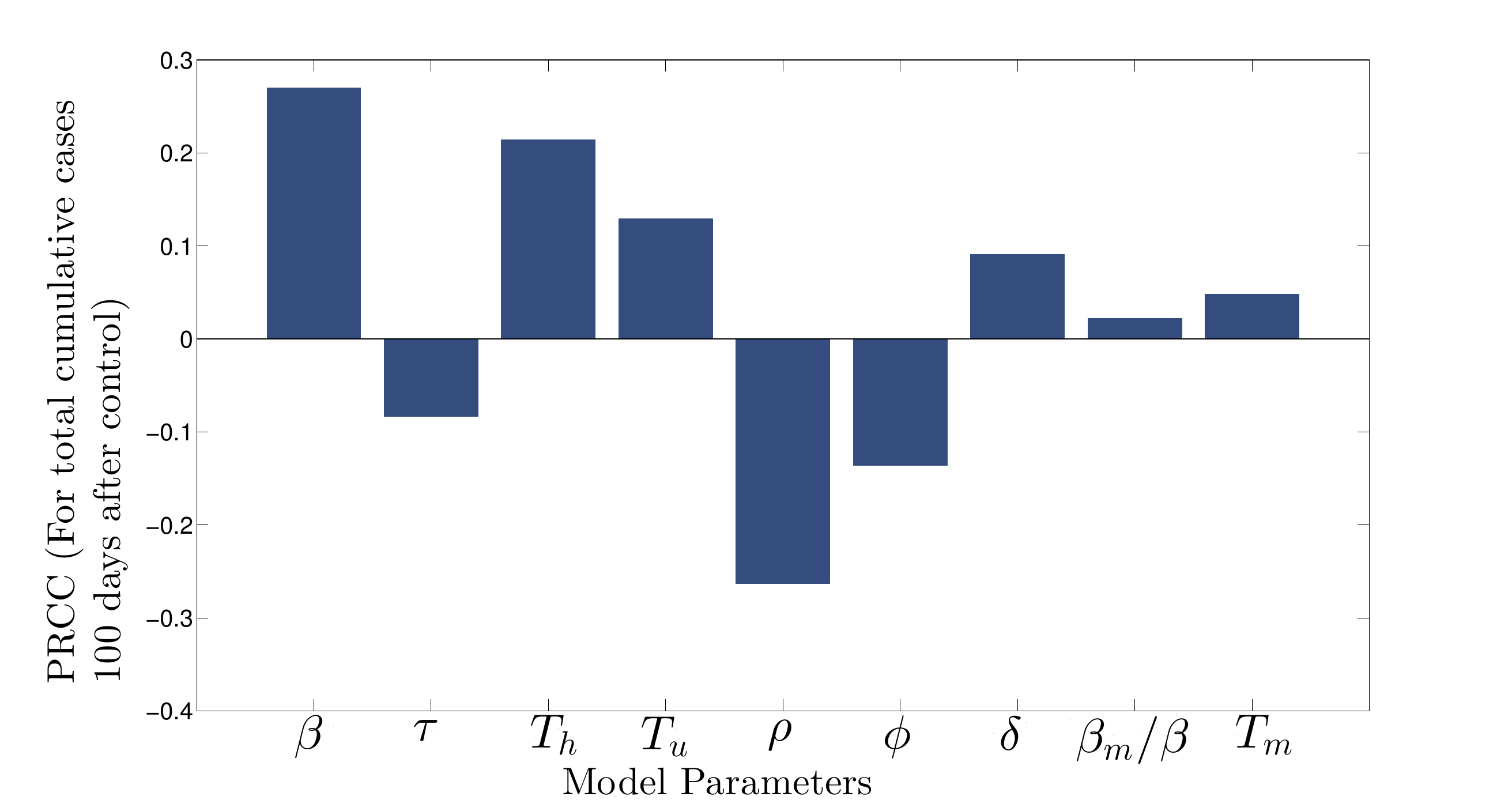}}
 \caption{ (a)   Global sensitivity analysis of $\mathcal R_e$ with respect to parameters using Latin Hypercube Sampling (LHS) and Partial-Rank Correlation Coefficients (PRCC). (a)   Sensitivity analysis of cumulative total cases with respect to parameters using LHS and PRCC.}
 \label{ReSens}
 \end{figure}  
\FloatBarrier

\bibliography{EbolaReferences}
\bibliographystyle{plain}

\end{document}